\documentclass[12pt]{article}
\pdfoutput=1
\usepackage[margin=1in]{geometry}

\usepackage[
  bookmarks=true,
  bookmarksnumbered=true,
  bookmarksopen=true,
  pdfborder={0 0 0},
  breaklinks=true,
  colorlinks=true,
  linkcolor=black,
  citecolor=black,
  filecolor=black,
  urlcolor=black,
]{hyperref}

\usepackage[utf8]{inputenc}
\usepackage{titlesec}
\usepackage{comment}
\usepackage{amsmath}
\usepackage{amsfonts}
\usepackage{amsthm}
\usepackage{subcaption}
\usepackage[round]{natbib}

\newtheorem{thm}{Theorem}[section]
\newtheorem{cor}[thm]{Corollary}

\newtheorem{lemma}[thm]{Lemma}

\newtheorem{property}[thm]{Property}

\theoremstyle{definition}
\newtheorem{definition}[thm]{Definition}
\newtheorem{example}[thm]{Example}

\newcommand\numberthis{\addtocounter{equation}{1}\tag{\theequation}}

\usepackage{algorithm}
\usepackage{algpseudocode}
\usepackage{parskip}
\usepackage{setspace}
\algnewcommand{\LineComment}[1]{\State \(\triangleright\) #1}

\DeclareMathOperator*{\argmax}{arg\,max}
\DeclareMathOperator*{\argmin}{arg\,min}

\title{Multi-Apartment Rent Division\thanks{Procaccia gratefully acknowledges research support by the National Science Foundation under grants IIS-2147187, IIS-2229881 and CCF-2007080; and by the Office of Naval Research under grant N00014-20-1-2488. Schiffer was supported by an NSF Graduate Research Fellowship. Zhang was supported by an NSF Graduate Research Fellowship.}}

\author{
    Ariel D. Procaccia\thanks{Paulson School of Engineering and Applied Sciences, Harvard University | \emph{E-mail}: \href{mailto:arielpro@seas.harvard.edu}{arielpro@seas.harvard.edu}.}
	\and
	Benjamin Schiffer\thanks{Department of Statistics, Harvard University | \emph{E-mail}: \href{mailto:bschiffer1@g.harvard.edu}{bschiffer1@g.harvard.edu}.}
	\and
	Shirley Zhang\thanks{Paulson School of Engineering and Applied Sciences, Harvard University | \emph{E-mail}: \href{mailto:szhang2@g.harvard.edu}{szhang2@g.harvard.edu}.}
}
\begin{document}

\begin{titlepage}
\maketitle

\setcounter{page}{0}
\thispagestyle{empty}

\begin{abstract}
    Rent division is the well-studied problem of fairly assigning rooms and dividing rent among a set of roommates within a single apartment. A shortcoming of existing solutions is that renters are assumed to be considering apartments in isolation, whereas in reality, renters can choose among multiple apartments. In this paper, we generalize the rent division problem to the multi-apartment setting, where the goal is to both fairly choose an apartment among a set of alternatives and fairly assign rooms and rents within the chosen apartment. Our main contribution is a generalization of envy-freeness called \textit{negotiated envy-freeness}. We show that a solution satisfying negotiated envy-freeness is guaranteed to exist and that it is possible to optimize over all negotiated envy-free solutions in polynomial time. We also define an even stronger fairness notion called \textit{universal envy-freeness} and study its existence when values are drawn randomly.
\end{abstract}

\end{titlepage}

\section{Introduction}

Rent division is a classic and intuitive problem within the space of fair division, in which a number of roommates are faced with the joint decision of how to assign rooms and split the total rent in a shared apartment. The problem is complicated by the fact that the rooms may vary widely and that players may have very different values for different rooms; for example, one room may have a better view, or one player may derive higher utility from having larger closets. Preferences over prices may also be complex, but it is commonly assumed that players' utilities are \emph{quasi-linear}: utility equals value minus price. The goal is to find an assignment of players and prices to rooms that takes such player preferences into account and satisfies a rigorous notion of fairness.  

The gold standard for fairness in rent division is \textit{envy-freeness}, which guarantees that each player has higher utility for their own room than for any other room, given the prices assigned to each room. In other words, in an envy-free allocation, no roommate would want to trade their room for any other room. Along with being easily justifiable~\citep{Pro19}, a solution which satisfies envy-freeness is also guaranteed to exist for the rent division problem \citep{Sven83} under the assumption of quasi-linearity. Furthermore, such a solution can be found in polynomial time \citep{Ara95}, and it is possible to optimize linear objectives over envy-free solutions in polynomial time as well \citep{GMPZ17}. Therefore, envy-freeness is both a compelling and computationally tractable fairness notion in the rent division setting.

The study of rent division is not driven merely by theoretical interest\,---\,it is a poster child for applications of fair division more broadly. In particular, the not-for-profit fair division website \emph{Spliddit}~\citep{GP14} operated between 2014 and 2022 and offered ``provably fair solutions to everyday problems.'' Its rent division application implemented the algorithm of \citet{GMPZ17}; among the five applications on Spliddit, it was the most popular, with more than 30,000 instances solved~\citep{PPZ22}.

A shortcoming of the prevalent approach to fair rent division, however, is that it assumes that the players have already chosen an apartment, and the only question remaining is how to divide the rooms and rent. By contrast, groups looking for an apartment are often not considering each potential apartment separately. Rather, many groups can further optimize by choosing among a set of available apartments which fit their budget and location constraints. This setting gives rise to new sources of complexity as the players are required not only  to assign rooms and prices in the chosen apartment, but also to collectively decide on which apartment to rent. Such a decision can be contentious. Imagine, for instance, three players whose apartment hunting priorities are a short commute, proximity to green spaces, and an exciting neighborhood, respectively. Suppose further that current apartment options include one apartment that is close to player $1$'s workplace, one which neighbors the largest park in the city, and one in the center of Restaurant Row. The rooms in each apartment may be asymmetric as well. Given all of these factors, how should the players decide which apartment to rent and who gets which room?

Crucially, it does not suffice to merely compute an envy-free solution in each individual apartment. Consider the following example:
\begin{example}\label{example:motivating}
    There are two players and two apartments. Each apartment has total rent $300$ and contains two symmetric rooms. The value of each player (rows) for each room (columns) is shown below:

    \parbox{.45\linewidth}{
        \centering
        \begin{tabular}{c | c c}
        & $r_{11}$ & $r_{12}$ \\
        \hline
        1 & 200 & 200 \\
        2 & 100 &   100  \\
        \end{tabular}
    }
    \parbox{.45\linewidth}{
        \centering
        \begin{tabular}{c | c c}
        & $r_{21}$ & $r_{22}$ \\
        \hline
        1 & 100 & 100 \\
        2 & 200 &   200  \\
        \end{tabular}
    }
\end{example}
In this example, the only solution which is individually envy-free in each apartment assigns equal rent in each apartment. However, the players will then disagree on which apartment to rent, as player $1$ will prefer apartment $1$ while player $2$ will prefer apartment $2$. It also seems intuitively unfair to assign equal rent in each apartment, as this would result in the two players having unequal utilities, even though their utility functions are symmetric. Therefore, a new fairness notion is needed for the multi-apartment rent division problem.

\subsection{Our Contributions}

First and foremost, we present a formal model of the rent division problem in the multi-apartment setting, and show that new fairness notions are necessary in this model. Our main contribution is a generalization of envy-freeness called \textit{negotiated envy-freeness}. We show that negotiated envy-freeness satisfies several desirable properties such as Pareto optimality and individual rationality, and reduces to envy-freeness in the single apartment setting. We provide an intuitive justification for negotiated envy-freeness based on negotiating rent between players who have different favorite apartments, and show how such negotiations allow players to reach a consensus apartment while maintaining fair rent burdens across players. We then show that a solution satisfying negotiated envy-freeness is guaranteed to exist and that it is possible to optimize linear objectives over all negotiated envy-free solutions in polynomial time, mirroring a similar result by \citet{GMPZ17} for the single-apartment setting. Finally, we introduce strong negotiated envy-freeness, a variant of negotiated envy-freeness which imposes additional fairness constraints on negotiations. 

We also explore \textit{universal envy-freeness}, which is the most direct generalization of envy-freeness. Unlike negotiated envy-freeness, however, we show that a universal envy-free solution is not guaranteed to exist (and, in fact, there are many instances where such a solution does not exist). Therefore, we instead study the probability of such a solution existing when players' utilities are drawn i.i.d. from an arbitrary distribution with a fixed number of players. For discrete distributions, we show that the probability that a universal envy-free solution exists approaches $1$ as the number of apartments approaches infinity. For continuous distributions, on the other hand, we show that the probability that a universal envy-free solution exists does not converge to either $0$ or $1$. However, if we add apartments one by one, the probability that there exists a stopping point where a universal envy-free solution exists converges to $1$, even for continuous distributions.

\subsection{Related Work}

As mentioned earlier, our work is most closely related to the paper of \citet{GMPZ17}. They study envy-free solutions for (single-apartment) rent division under quasi-linear utilities, and\,---\,building on earlier work by \citet{ADG91}\,---\,single out the \emph{maximin solution} (which maximizes the minimum utility subject to envy-freeness) as especially desirable. They also develop an algorithmic framework that allows for efficient computation of the maximin solution and a range of other objectives. 

By contrast, \citet{PPZ22} consider a different type of objective: they seek envy-free solutions that\,---\,rather than optimizing a welfare function\,---\,are robust to perturbations of the utilities. This objective is beyond the scope of our work.  

There are several approaches to rent division that relax the assumption of quasi-linear utilities or provide (incomparable) alternatives. One line of work allows players to express a (hard or soft) budget constraint~\citep{PVR18,AGGL+23,Vel22,Vel23}. \citet{ABR21} develop a fully-polynomial approximation scheme for envy-free rent division under the assumption that each player's utility (as a function of price) is continuous, monotone decreasing, and piecewise-linear. Finally, classic work by \citet{Su99} uses Sperner's Lemma to construct an algorithm for envy-free rent division under the ``miserly tenants'' assumption, which requires players to prefer a free room to any other room (even if its rent is \$1); while his approach is elegant, it has several shortcomings, including that preference elicitation requires repeated interaction with the players and that it is infeasible to optimize over envy-free solutions. \citet{Seg22}, however, shows that techniques developed for miserly tenants extend to the quasi-linear setting. All of these papers are orthogonal to ours, as we focus on the (widely used in practice) quasi-linear setting and instead extend the standard rent division problem to multiple apartments. 
	
It is worth noting that even in the basic setting of a single apartment and quasi-linear utilities, envy-freeness is incompatible with strategyproofness. For this reason, work on incentives in fair rent division is relatively limited. A notable exception is the work of \citet{Vel18}, who studies mechanisms whose equilibria give rise to envy-free solutions.

\section{Model}

A multi-apartment rent division instance is composed of a set of players $[n] = \{1,...,n\}$ and a set of apartments $[m] = \{1,...,m\}$, where each apartment $j$ consists of $n$ rooms $\{r_{j1},..., r_{jn}\}$. For each room $r_{jk}$ in apartment $j$, player $i$ has a non-negative value $V_i(r_{jk})$. Each apartment $j$ also has a total rent $R_j$. Unless otherwise noted, we will assume that $\sum_j \sum_k V_i(r_{jk}) = \sum_j R_j$ for all $i$. In other words, we do not assume that players have the same value for each apartment, but do assume that players have the same total value for all apartments under consideration. In our model, players are therefore able to express preferences over apartments, but their overall utility is still normalized as in the single apartment setting \citep{GMPZ17}. The entire instance can be represented by a valuation matrix $V \in \mathbb{M}_{n \times m \times n}(\mathbb{R}^+)$ and a rent vector $R \in \mathbb{R}^n$. 

An apartment assignment $A_j: [n] \to [n]$ is a mapping of players to rooms in apartment $j$, where $A_j(i)$ is the room assigned to player $i$ in apartment $j$. An assignment $A = \{A_1,...,A_m\}$ is the vector of mappings for all apartments. An assignment for apartment $j$ is \textit{welfare-maximizing} if it maximizes $\sum_{i = 1}^n V_i(A_j(i)) - R_j$ over all possible assignments in $j$. Rent is allocated via the price matrix $P \in \mathbb{R}^{m \times n}$, and we denote the price of a specific room $r$ as $P(r)$. We require that $\sum_{k \in [n]} P(r_{jk}) = R_j$ for any valid price matrix, and will refer to the price vector for apartment $j$ as $P_j$, where $P_j$ is the $j$th row of $P$. For a specific $A_j, P$, the quasi-linear utility of player $i$ is $U_i(A_j, P) = V_i(A_j(i)) - P(A_j(i))$. A \emph{solution} for the multi-apartment rent division instance is a tuple $(A, P, j^*)$, where $A$ and $P$ contain all assignments and prices and $j^*$ denotes the chosen apartment. We will often refer to a \emph{partial solution} $(A, P)$ where the apartment has not yet been chosen.

When there is only a single apartment ($m=1$), a solution $(A_1, P)$ is \textit{envy-free} (EF) if each player prefers her room over every other room. Formally, a solution is envy-free if for all $i, i' \in [n]$,
\begin{equation}\label{eq:one_apart}
    V_i(A_1(i)) - P(A_1(i)) \geq V_i(A_1(i')) - P(A_1(i')).
\end{equation}

In the single apartment setting, an envy-free solution always exists and can be found in polynomial time \citep{GMPZ17}. The goal of this paper is to generalize the notion of envy-freeness in Equation \eqref{eq:one_apart} to the multi-apartment setting when $m > 1$. In the multi-apartment setting, one starting point is to simply enforce the single apartment definition of envy-freeness for each apartment separately, i.e. for all $i,i' \in [n], j \in [m]$,
\begin{equation}\label{eq:ind_envy_free}
    V_i(A_j(i)) - P(A_j(i)) \geq V_i(A_j(i')) - P(A_j(i')).
\end{equation}
We will refer to partial solutions $(A,P)$ that satisfy Equation \eqref{eq:ind_envy_free} as \textit{individually envy-free}. However, as we showed in the introduction, it may be impossible for the players to agree on the final apartment choice $j^*$, and therefore this does not lead to an obviously fair solution. 

\section{Universal Envy-Freeness}\label{sec:uef}
We first consider \textit{universal envy-freeness}, a natural generalization of envy-freeness which captures the spirit of the original definition. Informally, a universal envy-free assignment guarantees that no player will want to switch her room in the chosen apartment with any room in any apartment. We formalize this definition below.
\begin{definition}
    A solution $(A, P, j^*)$ is \textbf{universal envy-free (UEF)} if for all $i,i' \in [n]$ and $j \in [m]$,
    \[
        V_i(A_{j^*}(i)) - P(A_{j^*}(i)) \geq V_i(A_j(i')) - P(A_j(i')).
    \]
\end{definition}

Unfortunately, a universal envy-free solution does not always exist, as can be seen in Example \ref{example:motivating}. In that example, the only way for both apartments to be individually envy-free is for rent to be assigned evenly in each apartment; i.e. that the price of each room is $150$. However, if rent is assigned in this way, then one of the players will be envious in the final apartment: if apartment $1$ were chosen, then player $2$ would rather have either room in apartment $2$, while if apartment $2$ were chosen, then player $1$ would rather have either room in apartment $1$. Therefore, no universal envy-free solution exists for Example \ref{example:motivating}. It turns out that this is not a knife's edge example, and that there are many instances which have no universal envy-free solution. For example, there are instances with no universal envy-free solution even when we restrict all players to have the same total value for each apartment ($\sum_{k = 1}^n V_i(r_{jk}) = R_j \: \forall i,j$). We note, however, that it is easy to check whether a universal envy-free solution exists via a simple linear program [Appendix \ref{app:check_uef}]. 

\subsection{Probabilistic Universal Envy-Freeness}

In computational fair division, when researchers were faced with the unfortunate non-existence of envy-free allocations of indivisible goods, they have asked whether such solutions are \emph{likely} to exist in random instances, at least in the large~\cite{DGKP+14,MS17,MS20}. In this section, we adopt the same approach in the context of universal envy-freeness. 

Specifically, we assume that the valuation matrix $V$ is drawn randomly, where each player's value for each room is drawn from a distribution $\mathcal{D}$ supported on $[0,1]$. In other words, for any player $i$ and room $r_{jk}$ in apartment $j$, the value of player $i$ for room $r_{jk}$ is $V_{i}(r_{jk}) \stackrel{i.i.d.}{\sim} \mathcal{D}$. Because each apartment is drawn symmetrically, we will also assume for this section only that the rent in every apartment is equal to $R$. Importantly, note that under this model we are no longer requiring player utilities to be normalized to add up to the total rent. However, as values are drawn i.i.d, every player will have the same total value for all of the rooms in expectation. 

Define $E_m$ as the event that there exists a universal envy-free solution with $m$ apartments when values are drawn i.i.d from distribution $\mathcal{D}$. First, we consider the simpler case when $\mathcal{D}$ is a discrete distribution with a finite number of values (see Appendix \ref{app:uef_proofs} for the proof).
\begin{thm}\label{lemma:discrete_UEF}
    If $V_{i}(r_{jk}) \stackrel{i.i.d.}{\sim} \mathcal{D}$ for all $i,j,k$ where $\mathcal{D}$ is a discrete distribution that takes on $\kappa < \infty$ distinct values, then for any constant $n$, $\Pr\left(E_m\right)  \xrightarrow[m \to \infty]{} 1$.
\end{thm}
The above result holds in the limit as $m$ goes to infinity, which is unrealistic in the real-world setting of searching for an apartment. However, a group of roommates may consider dozens of apartments in their search, in which case limit laws such as this result can become useful approximations. We would also ideally like to generalize this result to continuous distributions. Somewhat surprisingly, the same result does not hold in the case of continuous distributions supported on $[0,1]$. In fact for sufficiently large $m$, we can bound the probability that there exists a universal envy-free solution away from both $0$ and $1$ for a fixed value of $n$. Note that the lower bound is the same for any continuous distribution, while the upper bound is distribution-dependent. 

\begin{thm}\label{thm:probabilistic_uef}
  Suppose that $V_{i}(r_{jk}) \stackrel{i.i.d.}{\sim} \mathcal{D}$ for all $i,j,k$ where $\mathcal{D}$ is a continuous distribution supported on $[0,1]$. Then there exists a $p_0(n) > 0$ such that for any $m$ and $\mathcal{D}$,  $\Pr\left(E_m\right) \ge p_0(n)$. Furthermore, for any distribution $\mathcal{D}$, there exists a $p_1(n) < 1$ such that for any $m \ge n+1$, $\Pr\left(E_m\right) \le p_1(n)$.
\end{thm}
\begin{proof}[Proof sketch]
    We provide a brief sketch of the proof of Theorem \ref{thm:probabilistic_uef} and defer the formal proof to Appendix \ref{app:uef_proofs}. The structure of the proof is to construct two events $\mathcal{F}$ and $\mathcal{E}$ (both with probabilities independent of $m$) such that under event $\mathcal{F}$ a universal envy-free solution always exists and under event $\mathcal{E}$ a universal envy-free solution never exists. The key idea behind both constructed events is to consider the highest welfare achievable by potentially non-bijective assignments that can assign multiple rooms within an apartment to the same player. Note that these are not necessarily valid regular assignments. We will denote the maximum total utility achievable by a potentially non-bijective assignment for a given apartment $j$ as $MUW(j)$. Suppose the apartments are numbered in order of $MUW$, and therefore apartment $1$ has the highest $MUW$. We define $\mathcal{F}$ as the event when the potentially non-bijective assignment with the highest total utility in apartment $1$ is actually bijective, and therefore is a valid regular assignment. Under event $\mathcal{F}$, a universal envy-free solution exists (Lemma \ref{lemma:muw}). We then show that $\Pr(\mathcal{F})$ is always positive and independent of $m$ (Lemma \ref{lemma:muw_prob}). To show non-existence of a universal envy-free solution, we construct the event $\mathcal{E}$ with the three following conditions on the first $n + 1$ apartments based on the $MUW$ ordering. The event $\mathcal{E}$ occurs when apartment $n+1$ is the unique maximum welfare apartment, the potentially non-bijective assignment which achieves $MUW(n+1)$ is a bijective assignment, and for every $j \in [n]$, the potentially non-bijective assignment which achieves $MUW(j)$ assigns every room to player $j$ (Lemma \ref{lemma:nouef}). Conditioned on event $\mathcal{E}$, we show that no universal envy-free solution exists. Once again, $\Pr(\mathcal{E})$ is always positive and independent of $m$.
\end{proof}

Theorem \ref{thm:probabilistic_uef}  implies that simply starting with a very large number of apartments is not sufficient for guaranteeing existence of a universal envy-free solution. However, the construction of event $\mathcal{F}$ in Theorem \ref{thm:probabilistic_uef} also implies the following. Suppose there are $n$ players trying to find a universal envy-free solution, starting with $m_0$ apartments. If new apartments are added one at a time, and we check for universal envy-freeness before each new apartment is added, then with probability $1$ this process will terminate with a universal envy-free solution in a finite number of apartments. The proof of this result leverages the fact that event $\mathcal{F}$ relies only on the ordering of utilities within the apartment with the highest value of $MUW$. This result is formally outlined in Corollary \ref{cor:probabilistic_uef} and proven in Appendix \ref{app:uef_proofs}.
\begin{cor}\label{cor:probabilistic_uef}
    Suppose there exists an infinite sequence of apartments with $V_i(r_{jk})$ drawn i.i.d from a continuous distribution $\mathcal{D}$. Then for any constant $m_0 > 0$,
    \[
        \Pr\left( \inf \left\{m \ge m_0 : E_m \right\} < \infty \right) = 1.
    \]
\end{cor}

Conceptually, this corollary sends an encouraging message to apartment hunters: perseverance will likely lead to a fair outcome!

\section{Consensus and Negotiated Envy-Freeness}\label{sec:REF}

As universal envy-freeness is often too strong of a requirement, we would like to find a weaker condition that is always feasible, but still acts as a natural extension of envy-freeness in the single apartment setting. In this section, we introduce a fairness condition for the multi-apartment rent division problem that satisfies multiple desirable properties and is always guaranteed to exist. 

\subsection{Definition and Motivation}

First, we observe that a desirable condition in the multiple apartment setting is for all players to agree on the chosen apartment. In order for this to happen, every player must be at least as happy with their assignment in the chosen apartment as with their assignment in any other apartment. If this is the case, then we say that the chosen apartment is a \textit{consensus} apartment.

\begin{definition}
    A solution $(A, P, j^*)$ satisfies \textbf{consensus} if every player weakly prefers their assignment in $j^*$ to their assignment in any other apartment $j \neq j^*$. Formally, for every player $i$,
    \begin{equation}\label{eq:consensus}
        V_i(A_{j^*}(i)) - P(A_{j^*}(i)) \geq \max_{j \in [m]} V_i(A_{j}(i)) - P(A_{j}(i)). 
    \end{equation}
   A partial solution $(A,P)$ satisfies consensus if there exists an apartment $j^*$ satisfying Equation \eqref{eq:consensus}. We will refer to such a $j^*$ as a \textbf{consensus apartment} for $(A,P)$. 
\end{definition}

Given a partial solution that satisfies consensus, the set of consensus apartments is exactly the set of apartments which have the highest sum of player utilities. If there is more than one consensus apartment, every player has the same utility for their assigned room in each consensus apartment; the next lemma, whose proof is in Appendix \ref{app:highest_max_welfare}, formalizes this.

\begin{lemma}\label{lemma:highest_max_welfare}
    Let $(A, P)$ be a partial solution satisfying consensus. Apartment $j$ is a consensus apartment for $(A, P)$ if and only if 
    \[
        \sum_{i = 1}^n V_i(A_j(i)) - R_j \geq \sum_{i = 1}^n V_i(A_{j'}(i)) - R_{j'}
    \]
    for all $j' \neq j$. Furthermore, if $j_1$ and $j_2$ are both consensus apartments for $(A,P)$, then for all $i$, $U_i(A_{j_1}, P) = U_i(A_{j_2}, P)$.
\end{lemma}

Intuitively, consensus is a desirable property because a lack of consensus would imply that the players cannot decide on which apartment to rent. However, consensus alone is not sufficient. In the motivating example, for instance, a possible solution that satisfies consensus is to have player $1$ pay $300$ for room $r_{11}$ and $200$ for room $r_{21}$. However, this is clearly not a fair assignment, as player $1$ has $-100$ utility in both apartments while player $2$ has $100$ utility in both apartments, despite their valuations being perfectly symmetrical. This example can be further extended to make player $1$ arbitrarily unhappy in the consensus apartment. Therefore, we need a requirement stronger than just consensus to guarantee reasonable fairness for all players.

To prevent one player being significantly more unhappy in every apartment as in the previous example, we can instead try to reach consensus from a ``fair" starting partial solution. One such natural starting point is a partial solution $(A,P)$ which is individually envy-free. If this starting partial solution satisfies consensus, then the solution with a consensus apartment already satisfies universal envy-freeness. However, this starting partial solution may not satisfy consensus, as in Example \ref{example:motivating}. If no $j^*$ exists such that $(A,P,j^*)$ satisfies consensus, then we would like to adjust the rents $P$ in a fair way such that the resulting solution $(A,P',j^*)$ satisfies consensus. One way to do this is by negotiating a ``fair" compromise by conducting fair negotiations that change the prices $P$. An example of when negotiating prices may be useful is when apartment $j$ is player 1's favorite apartment and apartment $j'$ is player 2's favorite apartment. Then a negotiation between these two players and apartments could help balance the prices and either convince player $1$ to rent apartment $j'$ or convince player $2$ to rent apartment $j$, therefore bringing the partial solution closer to consensus.

We formalize this notion of negotiating as follows. Suppose we have a partial solution $(A,P)$ and a negotiation tuple $\tau = (\delta, i_1, i_2, j_1, j_2)$, where $\delta > 0$, $i_1, i_2 \in [n]$, $j_1, j_2 \in [m]$. Then a negotiation will consist of player $i_1$ increasing their rent in apartment $j_1$ by $\delta$ and decreasing their rent in apartment $j_2$ by $\delta$, while player $i_2$ decreases their rent in apartment $j_1$ by $\delta$ and increases their rent in apartment $j_2$ by $\delta$. Importantly, this negotiation only affects the players' rents and does not change the assignment $A$.  Formally, the partial solution  after making negotiation $\tau$ is $(A,P')$, where $P'(A_{j_1}(i_1)) = P(A_{j_1}(i_1)) + \delta$, $P'(A_{j_2}(i_1)) = P(A_{j_2}(i_1)) - \delta$, $P'(A_{j_1}(i_2)) = P(A_{j_1}(i_2)) - \delta$, $P'(A_{j_2}(i_2)) = P(A_{j_2}(i_2)) + \delta$, and for all other $(i,j)$ pairs, $P'(A_j(i)) = P(A_j(i))$. Note that each negotiation involves a player increasing rent in one apartment by $\delta$ and decreasing rent in another apartment by $\delta$, which enforces that the negotiations are fair. We define the partial solution  $(A,P)$ as \textit{reachable by negotiation }from the partial solution  $(A,Q)$ if there exists a series of $T$ negotiations  $\{(\delta^t, i_1^t, i_2^t, j_1^t, j_2^t)\}_{t=1}^T$ such that after making all $T$ negotiations starting from $(A,Q)$, the resulting partial solution is $(A,P)$. Returning to our concept of fairness, we want to consider partial solutions $(A,P)$ that are reachable by this form of fair negotiations from some individually envy-free starting partial solution $(A,Q)$. By construction of negotiations, the total utility of any player across all of their assigned rooms is the same in the final partial solution $(A,P)$ as in the individually envy-free starting partial solution $(A,Q)$. 

\begin{lemma}\label{obs:ref}
    A partial solution $(A,P)$ is reachable by negotiation from an individual envy-free starting partial solution $(A,Q)$ if and only if there exists an individually envy-free solution $(A,Q)$ such that for every player $i$, $ \sum_{j = 1}^m P(A_j(i)) = \sum_{j = 1}^m Q(A_j(i))$.
\end{lemma}

\begin{proof}[Proof sketch]
    The ``only if" direction follows from the fact that each negotiation does not change any player's total rent for all $m$ of their assigned rooms. Therefore, if there exists a sequence of negotiations that reach $(A,P)$ starting from $(A,Q)$, then every player has the same total rent for their $m$ assigned rooms in $(A,P)$ and $(A,Q)$.

    The ``if" direction requires a more technical construction, so we provide a brief proof sketch here and leave the formal proof to Lemma \ref{lemma:if_direction} in Appendix \ref{app:trading_equivalence}.  Suppose we start with an assignment $A$ and price matrices $Q$ and $P$ such that $ \sum_{j = 1}^m P(A_j(i)) = \sum_{j = 1}^m Q(A_j(i))$. We want to construct a series of negotiations to transform $Q$ into $P$. We do this for each apartment $1$ through $m$, one at a time. First, we construct a series of negotiations involving only apartments $1$ and $m$ such that the resulting price matrix $Q_1$ has the same prices as $P$ for apartment $1$. This is possible because the two price matrices have the same total rent within each apartment.
    We repeat this process for apartments $j \in \{2,...,m-1\}$ by constructing negotiations between apartment $j$ and apartment $m$ such that the resulting price matrix $Q_j$ has the same prices as $P$ for apartments $1$ through $j$. We then show that the final set of negotiations between apartment $m-1$ and $m$  results in a post-negotiation price matrix equal to the desired price matrix $P$. Therefore, $(A,P)$ is reachable by negotiation from $(A,Q)$.
\end{proof}
We now formally present the notion of fairness that comes from starting at an individually envy-free solution and conducting a sequence of negotiations until reaching consensus.  

\begin{definition}\label{def:REF}
    A solution $(A, P, j^*)$ satisfies \textbf{negotiated envy-freeness} if $(A, P, j^*)$ satisfies consensus and there exists a price matrix $Q$ such that $(A, Q)$ is individually envy-free and for every player $i$, 
    \[
        \sum_{j = 1}^m P(A_j(i)) = \sum_{j = 1}^m Q(A_j(i)).
    \]
\end{definition}

By Lemma \ref{obs:ref}, a solution that satisfies negotiated envy-freeness also satisfies that $(A,P)$ is reachable by negotiations from an individually envy-free partial solution $(A,Q)$.

Note that the final chosen apartment and prices in a solution satisfying negotiated envy-freeness do not have a meaningful fairness interpretation in isolation; that is, without the context of the original problem. Fundamentally, negotiated envy-freeness should be viewed as constructing a fair compromise when no single solution exists that every player prefers. The fairness of a compromise cannot (and should not) be evaluated without considering the full context, because all players must give something up in order to reach a compromise. For example, consider two friends A and B committed to renting an apartment together. To convince B to rent an apartment closer to A's workplace, A might offer the larger bedroom to B while paying equal rent, even though A prefers the larger bedroom. To an outside observer, it seems unfair that $A$ pays equal rent but gets the smaller bedroom. However, in the context of the original decision, this was a natural compromise. Similarly, the fairness of the negotiated envy-free solution cannot be evaluated by the chosen apartment only, but also needs to account for the original set of apartments.

For the rest of this section, we will study solutions $(A,P,j^*)$ that satisfy negotiated envy-freeness, and argue that this is a good generalization of single apartment envy-freeness. 

\subsection{Properties of Negotiated Envy-Freeness}
We will first show that a solution which satisfies negotiated envy-freeness also satisfies several desirable properties. In particular, such a solution satisfies Pareto optimality and individual rationality, and reduces to an envy-free solution in the single-apartment setting. Proofs of the below can be found in Appendix \ref{app:proof_of_REF_lemmas}. 

\begin{property}[Pareto optimality]\label{prop:pareto}
    A solution $(A, P, j^*)$ which satisfies negotiated envy-freeness also satisfies Pareto optimality, in that there exists no other solution $(A', P', j')$ such that $U_i(A'_{j'}(i), P') \geq U_i(A_{j^*}(i), P) \quad \forall i$ and the inequality is strict for at least one player. 
\end{property}

\begin{property}[Individual rationality]
\label{prop:ir}
    A solution $(A, P, j^*)$ which satisfies negotiated envy-freeness also satisfies individual rationality, in that all players have non-negative utility for their assigned rooms in the chosen apartment $j^*$. 
\end{property}

\begin{property}[Reduces to single-apartment setting]\label{prop:single}
    In the single-apartment setting (when $m=1$), a solution $(A, P, j^*)$ satisfies negotiated envy-freeness if and only if $(A, P, j^*)$ is envy-free within the single apartment.
\end{property}

\subsection{Existence of Negotiated Envy-Freeness}
We have shown that solutions which satisfy negotiated envy-freeness have both an intuitive explanation based on fair negotiating and desirable properties including Pareto optimality and individual rationality. Crucially, and in contrast to universal envy-freeness, it is also always possible to find a solution which satisfies negotiated envy-freeness.

\begin{thm}\label{thm:ref_existence}
    There exists a solution which satisfies  negotiated envy-freeness for every multi-apartment rent division instance.
\end{thm}
\begin{proof}[Proof sketch]
We provide a brief sketch of the proof of Theorem \ref{thm:ref_existence} and defer the formal proof to Appendix \ref{app:proof_of_ref_existence}. To prove this result, we will construct  a solution $(A,P^*,j^*)$ that satisfies negotiated envy-freeness for any instance of the problem. To do this, we first start with a  partial solution $(A,Q)$ that is individually envy-free within each apartment. Such a solution is guaranteed to exist. We then consider the solution $(A,P)$, where $P$ is the price matrix that,  for each apartment, induces equal utilities for all players in that apartment.  $(A, P)$ satisfies consensus, as the apartment which gives the highest utility for any player will be a consensus apartment. While $(A,P)$ may not satisfy negotiated envy-freeness, we can redistribute the prices evenly in $(A,P)$ through negotiating to get a new solution $(A,P^*)$ that satisfies the following two properties. First, the utility of a given player for each apartment in $(A, P^*)$ will be the same as in $(A, P)$, except that the utility may be scaled up or down by the same additive factor for all apartments. Second, we have $ \sum_{j = 1}^m P^*(A_j(i)) = \sum_{j = 1}^m Q(A_j(i))$. By this construction, we can conclude by showing that $(A,P^*)$ is a valid price assignment and satisfies negotiated envy-freeness and consensus.
\end{proof}

\subsection{Polynomial-Time Optimization}

We have shown that there always exists a solution which satisfies  negotiated envy-freeness, and the proof provides a constructive way to find such a solution in polynomial time. However, there may be many solutions which satisfy negotiated envy-freeness. In the single apartment setting, there exists a polynomial-time algorithm that optimizes a linear objective function over all envy-free solutions in polynomial time \citep{GMPZ17}. This raises the question of whether it is also possible to find a solution that optimizes a linear objective function in polynomial time over all negotiated envy-free solutions. As in the single-apartment setting, objective functions of special interest include maximin, which maximizes the utility of the least happy player, and equitability, which minimizes the disparity between players' utilities. Our main theorem of this section, Theorem \ref{thm:max_obj}, generalizes Theorem 3.1 of \citet{GMPZ17} to the multi-apartment setting with negotiated envy-freeness. 

\begin{thm}\label{thm:max_obj}
    Let $f_1,...,f_t: \mathbb{R}^{n \times m} \to \mathbb{R}$ be linear functions, where $t$ is polynomial in $n$ and $m$. Given a multi-apartment rent division instance, a solution $(A, P, j^*)$ that maximizes the minimum of $f_q(U_1(A_{j^*},P),...,U_n(A_{j^*},P))$ over all $q \in [t]$ subject to negotiated envy-freeness can be computed in time polynomial in both $n$ and $m$.
\end{thm}

Under this formalization, the maximin objective function can be represented by the linear functions $f_i(U_1(A_{j^*},P),...,U_n(A_{j^*},P)) = U_i(A_{j^*}, P)$ for all $i \in [n]$, and equitability can be represented by  $f_{i,i'}(U_1(A_{j^*},P),...,U_n(A_{j^*},P)) = U_i(A_{j^*}, P) - U_{i'}(A_{j^*}, P)$  for all $i, i' \in [n]$. 

While we defer the proof of Theorem \ref{thm:max_obj} to Appendix \ref{app:thm:max_obj}, we will state the key lemma used in the proof (Lemma \ref{lemma:multi_second_welfare}) and the algorithm that achieves the result.

Informally, Lemma \ref{lemma:multi_second_welfare} states that the set of solutions which satisfy negotiated envy-freeness are equivalent for all welfare-maximizing assignments, in the sense that the same set of player utilities can always be found for any choice of welfare-maximizing assignment. Note that Lemma \ref{lemma:multi_second_welfare} has a similar flavor to the 2nd Welfare Theorem \citep{GMPZ17,mas1995microeconomic} (see Lemma \ref{lemma:second_welfare}). 
\begin{lemma}\label{lemma:multi_second_welfare}
    Let $A, A'$ be two assignments that maximize welfare in every apartment, and let $P$ be a price matrix such that $(A, P, j^*)$ satisfies negotiated envy-freeness. Then there exists a price matrix $P'$ such that $(A', P', j^*)$ satisfies negotiated envy-freeness and $U_i(A_j, P) = U_i(A_j', P')$ for all $i,j$. 
\end{lemma}
The proof of Lemma \ref{lemma:multi_second_welfare} can be found in Section \ref{app:multi_second_welfare}. The algorithm used to prove Theorem \ref{thm:max_obj} is the following.
    \begin{algorithm}[tb]
    \caption{Optimizing an objective function subject to Negotiated Envy-Freeness}\label{algo:REF_consensus_opt}
    \begin{algorithmic}
    \For{$j \leftarrow 1$ to $m$}
        \State $A_j \gets$ welfare-maximizing assignment in $j$ via max-weight bipartite matching
    \EndFor
    \State $A \gets \{A_1,...,A_m\}$
    \State Choose $j^*_A$ such that $\sum_{i = 1}^n V_i(A_{j_A^*}(i)) - R_{j_A^*}\geq \max_j \sum_{i = 1}^n V_i(A_{j}(i)) - R_{j}$
    \begin{align*}
        P^*,Q^* &\gets \arg \max_{P,Q} \quad  Z \\
        \text{s.t.} \quad & Z \leq f_q(U_1(A_{j'},P),...,U_n(A_{j'},P)) \\
        & V_i(A_{j'}(i)) - P(A_{j'}(i)) \geq V_i(A_{j}(i)) - P(A_{j}(i)) \\
        & V_i(A_j(i)) - Q(A_j(i)) \geq V_i(A_j(i')) - Q(A_j(i'))  \\
        & \sum_j P(A_j(i)) = \sum_j Q(A_j(i)) \\
        & \sum_i P(A_j(i)) = R_j 
    \end{align*}
    \Return $(A, P^*, j^*)$ 
    \end{algorithmic}
    \end{algorithm}

\subsection{Strong Negotiated Envy-freeness}

We have so far defined two notions of fairness for multi-apartment rent division: universal envy-freeness and negotiated envy-freeness. Universal envy-freeness is simpler to interpret, but may not always exist. On the other hand, negotiated envy-freeness always exists, but there may be envy in the chosen apartment. In this section, we briefly describe an extension of negotiated envy-freeness (strong negotiated envy-freeness) that provides an interpretation for envy within the chosen apartment. Due to space constraints we defer the formal definition and proofs to Appendix \ref{sec:strong_REF}.

The definition of strong negotiated envy-freeness is similar to negotiated envy-freeness in that the solution $(A,P,j^*)$ must satisfy consensus and there must be an individually envy-free solution $(A,Q)$ that is reachable by negotiations. However, strong negotiated envy-freeness has an additional requirement that bounds the price differences between $P$ and $Q$ in the consensus apartment $j^*$. Due to this additional requirement, any envy in the final chosen apartment can be interpreted as being ``necessary" to reach consensus. We also show in Theorem \ref{thm:explainable} that a solution satisfying strong negotiated envy-freeness always exists and that optimizing an objective subject to strong negotiated envy-freeness can be done in polynomial time (as in Theorem \ref{thm:max_obj}).

\section{Discussion}

\subsection{Limitations}

Our model suffers the same limitations as that of \citet{GMPZ17}, including that envy-freeness is not strategy-proof and that the quasi-linear utility model is a strong simplifying assumption. We also acknowledge that negotiated envy-freeness, despite having a reasonable justification as reachable by negotiations, is less easily explainable as a fairness notion than envy-freeness, and that envy within the consensus apartment may still lead to discontent. From an application standpoint, it may be useful to let users find a negotiated envy-free solution themselves, by presenting users with individually envy-free solutions in each apartment and enabling the group to negotiate amongst themselves in the way dictated in the paper.

\subsection{Extensions}

We have required throughout this work that solutions to the multi-apartment rent division problem are of the form $(A, P, j^*)$, where $j^*$ is the chosen apartment. A natural generalization would be a distributional solution of the form $(A, P, D)$, where $D$ is a distribution over all $m$ apartments. We would then want to find a distributional solution $(A,P,D)$ such that in expectation, every player prefers their assignment to any other player's assignment. This notion of \textit{distributional envy-freeness} is a weaker notion of fairness than universal envy-freeness, and more instances of the problem have a distributional envy-free solution than a universal envy-free solution (including Example \ref{example:motivating}). See Appendix \ref{app:distributional} for more details on distributional envy-freeness.

In our work, we assumed that each group of $n$ players is determined to room together in an apartment with $n$ bedrooms. One extension is to have various sizes of apartments and allow the group of players to be split into multiple smaller apartments. This setting can be modeled as a cooperative game with transferable utility, and a natural question is then whether the core is always non-empty. In Appendix \ref{app:core}, we show that even when there are infinite copies of each apartment, there still exist instances where the core is empty. 

Another interesting question is how the solution changes as additional apartments are added (e.g. appear on the market). An additional apartment could change the set of negotiated envy-free solutions, which could in turn change the chosen apartment when optimizing objective functions such as maximin or equitability. We show in Appendix \ref{app:monotonicity} that the maximin solution under negotiated envy-freeness does not satisfy apartment monotonicity, in that an additional apartment could either raise or lower the achievable maximin value over all negotiated envy-free solutions.

In Section \ref{sec:uef}, we studied the existence of universal envy-free solutions when the values are drawn i.i.d. at random. In Appendix \ref{app:uef_corr}, we study the probability of the existence of a universal envy-free solution when there are two players with binary valuations. As the correlation between player values decreases, the probability of the existence of a universal envy-free solution generally increases. Interestingly, however, the increase is not monotonic. Therefore, sometimes the highest likelihood that there does not exist a universal envy-free solution occurs with correlation strictly between $-1$ and $1$, implying that correlation between players is not strictly better for finding a universal envy-free solution.

\section{Acknowledgements}

Procaccia gratefully acknowledges research support by the National Science Foundation under grants IIS-2147187, IIS-2229881 and CCF-2007080; and by the Office of Naval Research under grant N00014-20-1-2488. Schiffer was supported by an NSF Graduate Research Fellowship. Zhang was supported by an NSF Graduate Research Fellowship.

\newpage
\bibliographystyle{abbrvnat}
\bibliography{bib,abb,ultimate}

\newpage
\appendix
\section{Proof of Probabilistic Universal Envy-free Results}\label{app:uef_proofs}

\subsection{Proof of Theorem \ref{lemma:discrete_UEF}}
\begin{proof}
    Define $H$ as the maximum value of $\mathcal{D}$ and let $p$ be the probability that a random draw from $\mathcal{D}$ equals $H$. Define $E_H^j$ as the event that there exists an assignment $A^*_{j}$ for apartment $j$ such that for all players $i \in [n]$, $V_{i}(A_{j}(i)) = H$. For any $j$, $\Pr(E_H^j) \geq p^n$, and therefore $\Pr(\bigcup_j E_H^j) \ge 1-\left(1-p^n\right)^m$. Now we will show that $E_H^j$ implies that there exists a universal envy-free solution with apartment $j$ as the consensus apartment. Suppose we choose price matrix $P$ such that the price of every room in every apartment is equal to $R/n$. Then the solution $(A,P,j)$ where $A_{j} = A_j^*$ satisfies universal envy-freeness. This is because every player in apartment $j$ will have utility $H - \frac{R}{n}$, and no player can have more than utility $H - \frac{R}{n}$ for any room in any apartment because $H$ is the maximum of the distribution $\mathcal{D}$. Therefore, we have shown that $\Pr(E_m) \ge \Pr(\bigcup_j E_H^j) \ge 1-\left(1-p^n\right)^m   \xrightarrow[m \to \infty]{}  1$.
\end{proof}

\subsection{Proof of Theorem \ref{thm:probabilistic_uef}}

Before proving Theorem \ref{thm:probabilistic_uef}, we must introduce some additional notation that will allow us to define events under which a universal envy-free solution does and does not exist. Informally, we define the maximum unbalanced welfare of an apartment as the maximum possible welfare of that apartment if we could assign multiple rooms to the same player.

\begin{definition}\label{def:muw}
    Define the \textit{maximum unbalanced welfare (MUW)} of an apartment j as 
    \[
        MUW(j) = \left(\sum_{k} \max_{i \in [n]} V_i(r_{jk})\right) - R.
    \]
\end{definition}

In order to prove the lower bound in Theorem \ref{thm:probabilistic_uef}, we want to define an event $\mathcal{F}$ under which a universal envy-free solution exists for any number of apartments $m$. Furthermore, we want to be able to lower bound the probability of $\mathcal{F}$ by $p_0(n)$ that does not depend on $m$ or $\mathcal{D}$.

\begin{definition}\label{def:event_mathcal_F}
     Let  $j^*$ be the smallest $j$ such that $MUW(j) \ge MUW(j')$ for all apartments $j'$. Define $\mathcal{F}$ as the event that there exists an assignment $A^*_{j^*}$ such that
    \begin{equation}\label{eq:muw_assum}
        \sum_{i \in [n]} V_i(A^*_{j^*}(i)) - R = MUW(j^*).
    \end{equation}
\end{definition}

\begin{lemma}\label{lemma:muw}
    Conditioned on event $\mathcal{F}$, there exists a universal envy-free solution.
\end{lemma}

    \begin{proof}
    Construct assignment $A$ as follows. Let $A_{j^*} = A^*_{j^*}$. For every other apartment $j$, let $A_j$ be a welfare-maximizing assignment for apartment $j$. Now, construct the price matrix $P$ as follows. Let
    \[
        P\left(A_{j^*}(i)\right) = V_i(A_{j^*}(i)) - \frac{MUW(j^*)}{n} \quad \forall \: i.
    \]
    Note that these are valid prices due to Equation \eqref{eq:muw_assum}. A consequence of this is
    \begin{equation}\label{eq:equal_util}
        U_i(A_{j^*}, P) = V_i(A_{j^*}(i)) - P\left(A_{j^*}(i)\right)  =  \frac{MUW(j^*)}{n} \quad \forall \: i.
    \end{equation}
    
    For every other apartment $j \ne j^*$, assign prices in a way that minimizes the maximum utility any player has for any room within that apartment. Formally, let $\mathcal{P}_j$ be the set of all possible prices for apartment $j$. Then
    \[
        P_j = \argmin_{P'_j \in \mathcal{P}_j} \max_{i,k \in [n]} V_i(r_{jk}) - P'_j(r_{jk}) \quad \forall \: j \neq j^*.
    \]
    Under this choice of prices, for any apartment $j \ne j^*$ and rooms $r_{jk_1},r_{jk_2}$ in apartment $j$,
    \[
        \max_{i \in [n]} V_i(r_{jk_1}) - P(r_{jk_1}) =  \max_{i \in [n]} V_i(r_{jk_2}) - P(r_{jk_2}).
    \]
    By Definition \ref{def:muw}, $MUW(j) = \sum_{k \in [1:n]} \max_{i \in [n]} (V_i(r_{jk}) - P(r_{jk}))$. Therefore, for every room $r_{jk}$ in apartment $j \ne j^*$, 
    \begin{equation}\label{eq:maxval}
        \max_{i \in [n]} V_i(r_{jk}) - P(r_{jk}) = \frac{MUW(j)}{n}.
    \end{equation}
    
    We will now show that $(A, P, j^*)$ satisfies universal envy-freeness. First, we will show that no player strictly prefers another room in apartment $j^*$ to their assigned room in apartment $j^*$. For any players $i, i'$,
    \begin{align*}
        V_i(A_{j^*}(i)) - P(A_{j^*}(i)) &= V_{i'}(A_{j^*}(i')) - P(A_{j^*}(i')) \\
        &\geq V_{i}(A_{j^*}(i')) - P(A_{j^*}(i')) \numberthis \label{eq:jstar_ef}\\
    \end{align*}
    The equality holds by an application of Equation \eqref{eq:equal_util}. The inequality holds by Equation \eqref{eq:muw_assum} and Definition \ref{def:muw}, which together imply that $V_{i'}(A_{j^*}(i')) \geq V_{i}(A_{j^*}(i'))$.
        
    To finish showing that $(A,P,j^*)$ is a universal envy-free solution, we must show that no player prefers any room in any $j \neq j^*$ to their assigned room in apartment $j^*$. By Equation \eqref{eq:maxval}, 
    \[
        U_i(A_j,P) = V_i(A_j(i)) - P(A_j(i)) \leq \frac{MUW(j)}{n} \quad \forall \: j \neq j^*.
    \]
    Furthermore, by Equation \eqref{eq:equal_util}, $U_i(A_{j^*},P)  = \frac{MUW(j^*)}{n}$ for every player $i$. By assumption, $MUW(j^*) \geq MUW(j) \: \forall \: j$, and therefore for every player $i$ and apartment $j \ne j^*$,
    \[
        U_i(A_{j^*},P)  = \frac{MUW(j^*)}{n} \ge \frac{MUW(j)}{n} \ge U_i(A_j,P)
    \]
    Therefore, $(A,P,j^*)$ is a universal envy-free solution.
\end{proof}

\begin{lemma}\label{lemma:muw_prob}
    Suppose that $V_{i}(r_{jk}) \stackrel{i.i.d.}{\sim} \mathcal{D}$ for all $i,j,k$ where $\mathcal{D}$ is a continuous distribution supported on $[0,1]$. Then for all $m$, $\Pr\left(\mathcal{F}\right) \ge \frac{n!}{n^n}$.
\end{lemma}

\begin{proof}
    Let  $j^*$ be the smallest $j$ such that $MUW(j) \ge MUW(j')$ for all apartments $j'$. By construction, $\mathcal{F}$ only depends on the values in apartment $j^*$. Again by Equation \eqref{eq:muw_assum} and Definition \ref{def:muw}, $\mathcal{F}$ occurs if and only if there exists an assignment $A^*_{j^*}$ such that
    \[
        V_i(A^*_{j^*}(i)) \geq V_{i'}(A^*_{j^*}(i)) \quad \forall \: i,i'.
    \]
    Now consider the following process of generating an i.i.d. valuation matrix from $\mathcal{D}$. First, draw $n$ i.i.d. values for each room from $\mathcal{D}$. Then, map these $n$ values to the $n$ players uniformly at random.

    Note that whether event $\mathcal{F}$ occurs does not depend on the \textit{which} values are drawn for rooms in apartment $j^*$, but rather depends only on the \textit{mapping} of those values to players. To compute $\Pr(\mathcal{F})$, it suffices to compute the probability that for each room in apartment $j^*$, the maximum of the $n$ values is mapped to a different person. Let $\pi$ be a permutation from $[n]$ to $[n]$. Then the probability that, for all $i$, the maximum value of room $r_{j^*\pi(i)}$ is mapped to player $i$ is exactly $\frac{1}{n^n}$. Since there are $n!$ such permutations $\pi$, we can conclude that $\Pr(\mathcal{F}) = \frac{n!}{n^n}$.
\end{proof}

We now define an event $\mathcal{E}$ under which there is no universal envy-free solution. We show in Lemma \ref{lemma:nouef} that the probability of this event is bounded away from $0$.

\begin{definition}\label{def:event_mathcal_E}
    Let $m \ge n+1$ and $n \ge 2$. Assume w.l.o.g that the $m$ apartments $1,2,...,m$ are ordered by MUW, i.e. $MUW(j) \ge MUW(j')$ for $j \le j'$. Define $\mathcal{E}$ as the event that the following all hold:
    \begin{enumerate}
        \item $MUW(1) > MUW(2) > ... > MUW(m)$ (i.e. inequalities are strict).
        \item Define $A_j^*$ be a welfare maximizing assignment for apartment $j$. Then
        \[
             \sum_{i} V_i(A_{n+1}^*(i)) - R > \sum_{i} V_i(A_j^*(i)) - R   \quad \: \forall j.
        \]
        \item For every apartment $j \in [n]$, 
        \[
            V_j(r_{jk}) \ge V_i(r_{jk}) \quad \forall \: i,k \in [n].
        \]
        In other words, the $jth$ player has the highest value over all players for every room in apartment $j$.
    \end{enumerate}
\end{definition}

\begin{lemma}\label{lemma:nouef}
    Suppose that $m \ge n+1$ and $n \ge 2$ and event $\mathcal{E}$ holds. Then there is no universal envy-free solution. Furthermore, $\Pr(\mathcal{E}) > 0$ and $\Pr(\mathcal{E})$ depends only on $n$ and $\mathcal{D}$. 
\end{lemma}
\begin{proof}  
    We first show that there does not exist a universal envy-free solution of the form $(A, P, j^*)$ with $j^* \neq n+1$. Suppose we have a universal envy-free solution $(A,P,j^*)$ where $j^* \ne n+1$. Then by definition of universal envy-freeness, $U_i(A_{j^*}, P) \ge U_i(A'_{n+1},P)$ for every assignment $A'_{n+1}$ in apartment $n+1$. Let $A_{n+1}^*$ be a welfare-maximizing assignment in apartment $n+1$ and let $A_{j^*}^*$ be a welfare-maximizing assignment in apartment $j^*$. Then we must have that 
    \[
        \sum_i U_i(A^*_{j^*}, P) \ge \sum_i U_i(A_{j^*}, P)  \ge \sum_i U_i(A_{n+1}^*,P),
    \]
    however this is a contradiction with the second condition of $\mathcal{E}$.
    
    Now we will show that there are no universal envy-free solutions of the form $(A,P,n+1)$. If $(A,P,n+1)$ is a universal envy-free solution, then
    \[
        U_i(A_{n+1}, P) \ge U_i(A_{j}, P) \quad \forall j.
    \]
   By the third condition of $\mathcal{E}$, for any $i \in [n]$ and any price matrix $P$,
   \[
    \sum_{k} U_i(A_i, P) = \sum_{k} V_i(r_{ik}) - R = MUW(i).
   \]
   Therefore by the pigeonhole principle, there must be at least one room in apartment $i$ for which player $i$ has utility at least $\frac{MUW(i)}{n}$. Putting these last two equations together, if $(A,P,n+1)$ satisfies universal envy-freeness, then
    \begin{equation}\label{eq:utillowerbound}
        \frac{MUW(i)}{n} \leq U_i(A_{n+1}, P) \quad \forall i \in [n].
    \end{equation}
    By Definition \ref{def:muw}, it must be the case that
    \begin{equation}\label{eq:utilupperbound}
        \sum_{i=1}^n U_i(A_{n+1}, P) \le MUW(n+1) \quad  \forall i \in [n].
    \end{equation}
    Combining Equations \eqref{eq:utillowerbound} and \eqref{eq:utilupperbound} gives 
    \[
       \sum_{i=1}^n \frac{MUW(i)}{n} \leq \sum_{i=1}^n U_i(A_{n+1}, P) \leq  MUW(n+1).
    \]
    However, by the definition of event $\mathcal{E}$ we have $MUW(i) > MUW(n+1)$ for all $i$. Therefore, we have reached a contradiction and can conclude that under event $\mathcal{E}$, there is no universal envy-free solution.

    To lower bound the probability of $\mathcal{E}$, it will be helpful to consider the following process of generating an i.i.d. valuation matrix. We define $\mathcal{D}_{\mathrm{max}}$ as the distribution of the maximum of $n$ draws from $\mathcal{D}$. For each room $r_{jk}$ in each apartment $j$, first draw one value $v_{\mathrm{max}}(r_{jk}) \sim \mathcal{D}_{\mathrm{max}}$. Next, draw another $n-1$ values from $\mathcal{D}$ conditioned on each value being at most $v_{\mathrm{max}}(r_{jk})$. Finally, map these $n$ values to the $n$ players uniformly at random.
    
    Now, we will define a more complicated distribution-specific event $E = E_0 \cap E_1 \cap E_2 \cap E_3$ such that $E \subseteq \mathcal{E}$. We will lower bound $\Pr(E)$, which will immediately yield a lower bound for $\Pr(\mathcal{E})$. Define $\mu_U$ and $\mu_L$ as the maximum and minimum of the distribution $\mathcal{D}$, i.e. $\mu_U = \inf \left\{ x : \Pr_{X \sim \mathcal{D}} (X \le x) = 1\right\}$ and $\mu_L = \inf \left \{ x : \Pr_{X \sim \mathcal{D}} (X \le x) > 0\right\}$. Consider the following events $E_0, E_1, E_2$, and $E_3$:
    \begin{itemize}
        \item Define $E_0$ as the event that there are no two apartments $j$ and $j'$ such that $MUW(j') = MUW(j)$. Because $\mathcal{D}$ is continuous, $\Pr(E_0) = 1$.
        \item Define $E_1$ as the event that for every room $r_{jk}$ in apartments $j=1,2,...n+1$, the value of $v_{\mathrm{max}}(r_{jk})$ is greater than $\frac{3\mu_U + \mu_L}{4}$. 
        
        For any single draw from $\mathcal{D}_{\mathrm{max}}$, the probability of being greater than $\frac{3\mu_U + \mu_L}{4}$ is at least the probability that a random draw from $\mathcal{D}$ is greater than $\frac{3\mu_U + \mu_L}{4}$ because $\mathcal{D}_{\mathrm{max}}$ stochastically dominates $\mathcal{D}$. Therefore, the probability that $v_{\mathrm{max}}(r_{jk})$ is greater than $\frac{3\mu_U + \mu_L}{4}$ for a single room $r_{jk}$ is at least $\Pr_{x \sim \mathcal{D}}\left(x \ge \frac{3\mu_U + \mu_L}{4}\right)$. There are $n(n+1)$ rooms in the first $n+1$ apartments, and the maximum value for each room is independent across rooms, and therefore 
        \begin{equation}\label{eq:E1}
            \Pr(E_1) \ge \Pr_{x \sim \mathcal{D}}\left(x \ge \frac{3\mu_U + \mu_L}{4}\right)^{n(n+1)} > 0.
        \end{equation}
        \item Define $E_2$ as the event that in apartment $n+1$, there exists a permutation $\pi: [n] \to [n]$ such that for all $i \in [n]$, player $i$ has the maximum value for room $\pi(i)$. The probability of event $E_2$ is $\frac{n!}{n^n}$ as in the proof of Lemma \ref{lemma:muw}. Therefore
        \begin{equation}\label{eq:E2}
            \Pr(E_2) = \frac{n!}{n^n}.
        \end{equation}
        \item Define event $E_3$ as the event that for every apartment $j \in [1,...,n]$ and player $i \ne j$, the value of player $i$ for every room in apartment $j$ is at most $\frac{\mu_U + 3\mu_L}{4}$. 
        
        We will lower bound the probability of event $E_3$ conditioned on event $E_1$. Consider any fixed room $r_{jk}$. Conditioned on event $E_1$, the maximum value in room $r_{jk}$ is at least $\frac{3\mu_U + \mu_L}{4}$. Therefore, in order to have $V_i(r_{jk}) \le \frac{\mu_U + 3\mu_L}{4}$ for all $i \ne j$, player $j$ must have the highest value among all players for room $r_{jk}$. Furthermore, the other $n-1$ values for room $r_{jk}$ must all be  less than $\frac{\mu_U + 3\mu_L}{4}$. The other $n-1$ values for room $r_{jk}$ are drawn independently. Therefore, the probability that the other $n-1$ values for room $r_{jk}$ are all less than $\frac{\mu_U + 3\mu_L}{4}$ is
        {\footnotesize
        \begin{align*}
            \Pr_{x \sim \mathcal{D}}\left(x \le \frac{\mu_U + 3\mu_L}{4} \middle| x \le v_{\mathrm{max}}(r_{jk}), E_1\right)^{n-1}  &= \Pr_{x \sim \mathcal{D}}\left(x \le\frac{\mu_U + 3\mu_L}{4} \middle| x \le v_{\mathrm{max}}(r_{jk}), v_{\mathrm{max}}(r_{jk}) \ge \frac{3\mu_U + \mu_L}{4} \right)^{n-1}  \\
            &\ge  \Pr_{x \sim \mathcal{D}}\left(x \le \frac{\mu_U + 3\mu_L}{4}\right)^{n-1} \\
            &> 0.
        \end{align*}}
        
        Note that the probability that player $j$ has the highest value for room $r_{jk}$ among all $n$ players is $1/n$, and that this event is independent of event $E_1$. Putting this all together, we have 
        \[
            \Pr\left(V_i(r_{jk}) \le \frac{\mu_U + 3\mu_L}{4} \: \: \forall i \ne j \: \middle | \: E_1\right) \ge \frac{1}{n}\Pr_{x \sim \mathcal{D}}\left(x \le \frac{\mu_U + 3\mu_L}{4}\right)^{n-1}.
        \]
        There are $n^2$ such rooms $r_{jk}$ in the first $n$ apartments. Furthermore, for two rooms $r \ne r'$, the event that for $V_i(r) \le \frac{\mu_U + 3\mu_L}{4}$ for all $i \ne j$ is independent of the event that $V_i(r') \le \frac{\mu_U + 3\mu_L}{4}$ for all $i \ne j$. Therefore, we have that
        \begin{equation}\label{eq:E3}
            \Pr(E_3 | E_1) \ge \left(\frac{1}{n}\Pr_{x \sim \mathcal{D}}\left(x \le \frac{\mu_U + 3\mu_L}{4}\right)^{n-1}\right)^{n^2}.
        \end{equation}
    \end{itemize}

     Recall that $E = E_0 \cap E_1 \cap E_2 \cap E_3$. 
     
     Next, we will show that $E \subseteq \mathcal{E}$. The first condition of Definition \ref{def:event_mathcal_E} is satisfied under event $E$ by definition of event $E_0$. Recall as argued above that events $E_1$ and $E_3$ together imply that player $j$ must have the highest value among all players for every room in apartment $j$ for $j \in [n]$. This implies that under event $E$ the third condition of Definition \ref{def:event_mathcal_E} is satisfied. Finally, we will show that under event $E$, the second condition of Definition \ref{def:event_mathcal_E} is satisfied. Under events $E_1$ and $E_2$, there exists an assignment $A_{n+1}^*$ in apartment $n+1$ such that every player $i$ has value at least $\frac{3\mu_U + \mu_L}{4}$ for room $A_{n+1}^*(i)$. This implies that
     \begin{equation}\label{eq:welfare_n+1}
         \sum_{i=1}^n V_i(A_{n+1}^*(i)) - R \ge n \cdot \frac{3\mu_U + \mu_L}{4} - R.
     \end{equation}
     For every assignment $A_j$ for apartment $j \in [n]$, under event $E_3$, every player $i \ne j$ has value at most $\frac{\mu_U + 3\mu_L}{4}$ for room $A_j(i)$. Because $\mathcal{D}$ is bounded, player $j$ has value at most $\mu_U$ for room $A_j(j)$. Let $A_j^*$ be a welfare-maximizing assignment in apartment $j$. Putting this all together, we have that for any $j \in [n]$, conditioned on event $E_3$
     \begin{equation}\label{eq:welfare_j}
         \sum_{i=1}^n V_i(A_{j}^*(i)) - R \leq (n-1) \cdot \frac{\mu_U + 3\mu_L}{4} + \mu_U - R.
     \end{equation}
     Because $\mu_U > \mu_L$, we note that $(n-1) \cdot \frac{\mu_U + 3\mu_L}{4} + \mu_U <  n \cdot \frac{3\mu_U + \mu_L}{4}$ for $n \ge 2$. This combined with Equations  \eqref{eq:welfare_n+1} and \eqref{eq:welfare_j} implies that for any $j \in [n]$, conditioned on event $E$,
     \[
       \sum_{i=1}^n V_i(A_{j}^*(i)) - R  < \sum_{i=1}^n V_i(A_{n+1}^*(i)) - R.
     \]
     Now we need to show the same result for $j > n+1$. For any apartment $j$ and any assignment $A_j$, $\sum_{i=1}^n V_i(A_{j}(i)) - R  \le MUW(j)$. Furthermore, event $E_2$ implies that $\sum_{i=1}^n V_i(A_{n+1}^*(i)) - R = MUW(n+1)$. Let $A_j^*$ be a welfare-maximizing assignment in apartment $j$. Putting this together with event $E_0$ and that the apartments are ordered by MUW, we have that conditioned on events $E_0$ and $E_2$, for any $j > n+1$,
     \begin{align*}
         \sum_{i=1}^n V_i(A_{j}^*(i)) - R  &\le MUW(j) \\
         &< MUW(n+1) \\
         &=  \sum_{i=1}^n V_i(A_{n+1}^*(i)) - R.
     \end{align*}
     We have now exactly shown that the second condition for event $\mathcal{E}$ is satisfied conditioned on event $E$. Therefore, we have shown that $E \subseteq \mathcal{E}$. 
     
     Now all that remains to be done is to lower bound the probability of event $E$. Note that by construction, event $E_2$ is independent of events $E_1$ and $E_3$. Using this with the fact that $\Pr(E_0) = 1$ and Equations \eqref{eq:E1}, \eqref{eq:E2}, \eqref{eq:E3}, we have that
    \begin{align*}
        \Pr(E_0 \cap E_1 \cap E_2 \cap E_3) &= \Pr(E_2)\cdot \Pr(E_1)  \cdot \Pr(E_3 | E_1) \\
        &\ge \frac{n!}{n^n} \cdot \Pr_{x \sim \mathcal{D}}\left(x \ge \frac{3\mu_U + \mu_L}{4}\right)^{n(n+1)} \cdot\left(\frac{1}{n}\Pr_{x \sim \mathcal{D}}\left(x \le \frac{\mu_U + 3\mu_L}{4}\right)^{n-1}\right)^{n^2} \\
        &:= p(n).
    \end{align*}
    Therefore, we can also conclude that $\Pr(\mathcal{E}) \ge \Pr(E) \ge p(n)$ which only depends on the distribution $\mathcal{D}$ and $n$.

\end{proof}

\subsection{Proof of Corollary \ref{cor:probabilistic_uef}}

\begin{proof}
        Let $\ell_1,\ell_2,...$ be an a strictly increasing sequence of integers greater than $m_0$ such that $\argmax_{j \in [\ell_t]} MUW(j) = \ell_t$ for all $t$. Define $\mathcal{F}_t$ as the event from \ref{def:event_mathcal_F} considering only the first $\ell_t$ apartments. By definition, $\mathcal{F}_t$ only depends on the mapping of values in apartment $\ell_t$. Therefore event $\mathcal{F}_{t}$ and $\mathcal{F}_{t'}$ are independent for any $t \ne t'$. By Lemma \ref{lemma:muw}, $\Pr(\mathcal{F}_t) = \frac{n!}{n^n}$ for all $t$. Define $T$ as the random variable that is the smallest $t$ such that $\mathcal{F}_t$ holds. $T$ is a geometric random variable with probability $\frac{n!}{n^n}$, therefore with probability $1$, $T$ is finite. By definition of $\mathcal{F}_T$, there must exist a universal envy-free solution for apartments $1,..,\ell_T$. Putting this all together, with probability $1$, there will exist a universal envy-free solution for apartments $1,...,\ell_T$ for some finite $T$. 

        Now we need to show that for any finite $T$, with probability $1$ there exists a strictly increasing sequence $\ell_1,...,\ell_T$ such that $\argmax_{j \in [\ell_t]} MUW(j) = \ell_t$. Denote the distribution of the $MUW$ of an apartment with randomly drawn values as $\mathcal{M}$. As before, define $\mu_U= \inf \left\{ x : \Pr_{X \sim \mathcal{D}} (X \le x) = 1\right\}$. By definition, if $Z \sim \mathcal{M}$ then $Z \le n \cdot \mu_U-R$. Because $\mathcal{D}$ is continuous, $\mathcal{M}$ is also continuous, which implies $\Pr_{Z \sim \mathcal{M}}(Z < n \cdot \mu_U - R) = 1$. Furthermore, if $\epsilon > 0$, then $\Pr_{Z \sim \mathcal{M}}(Z > n \cdot \mu_U - R- \epsilon) > 0$. This implies that for any finite $T$,  with probability $1$ there will be a sequence $\ell_1,\ell_2,...,\ell_T$ such that $\ell_t = \argmax_{j \in [\ell_t]} MUW(j)$ for all $t \le T$. Combining this with the result in the first paragraph gives the desired result.
\end{proof}

\subsection{Checking Existence of a Universal Envy-free Solution}\label{app:check_uef}

We can check whether a given multi-apartment rent division instance has a universal envy-free solution using the following linear program. Note that the  linear program has no objective, as we are only checking for feasibility.
\begin{align*}
    & \max_{j^*, P} \mathbf{0} \\
    \text{s.t. } & \sum_i P(A_j(i)) = R_j & \forall \: j & \qquad \text{[payments equal rent]} \\
    & V_i(A_{j^*}(i)) - P(A_{j^*}(i)) \geq V_i(A_j(i')) - P(A_j(i')) & \forall \: i, i', j & \qquad \text{[universal EF]} \\
\end{align*}

\section{Proofs of Lemmas from Section \ref{sec:REF}}\label{app:proof_of_REF_lemmas}

\subsection{Proof of Lemma \ref{lemma:highest_max_welfare}}\label{app:highest_max_welfare}

Suppose that $(A, P)$ is a partial solution satisfying consensus. We first prove that consensus can only occur in an apartment $j$ if for all $j' \neq j$,
\begin{equation}\label{eq:max_welfare_assignment}
    \sum_{i = 1}^n V_i(A_j(i)) - R_j \geq \sum_{i = 1}^n V_i(A_{j'}(i)) - R_{j'}.
\end{equation}

Consider any apartment $j$ such that Equation \eqref{eq:max_welfare_assignment} does not hold. This implies that there is some apartment $j'$ such that $\sum_{i = 1}^n V_i(A_{j'}(i)) - R_{j'} > \sum_{i = 1}^n V_i(A_j(i)) - R_j$. This implies that 
\[
    0 < \sum_{i = 1}^n V_i(A_{j'}(i)) - R_{j'} - \sum_{i = 1}^n V_i(A_j(i)) - R_j =  \sum_{i = 1}^n U_i(A_{j'}(i)) - \sum_{i = 1}^n U_i(A_{j}(i)).
\]
By the pigeonhole principle, at least one player must have strictly higher utility in $j'$ than $j$. Therefore, $j$ cannot be a consensus apartment.

We now prove that if Equation \eqref{eq:max_welfare_assignment} holds for apartment $j$, then apartment $j$ must be a consensus apartment. Let $j^*$ be the consensus apartment under $(A,P)$. We will show that if apartment $j$ satisfies Equation \eqref{eq:max_welfare_assignment}, then for all $i$,
\begin{equation}\label{eq:contra}
    U_i(A_{j}, P) = U_i(A_{j^*}, P),
\end{equation}
which directly implies that $j$ must be a consensus apartment. Proof by contradiction. Assume that apartment $j$ satisfies Equation \eqref{eq:max_welfare_assignment} but does not satisfy Equation \eqref{eq:contra}.  Because $j^*$ is a consensus apartment, by the first part of this proof we have that $\sum_{i = 1}^n V_i(A_{j^*}(i)) - R_{j^*} \geq \sum_{i = 1}^n V_i(A_{j}(i)) - R_{j}$. This combined with the assumption that Equation \eqref{eq:max_welfare_assignment} holds for apartment $j$ implies that $\sum_{i = 1}^n V_i(A_{j}(i)) - R_{j} = \sum_{i = 1}^n V_i(A_{j^*}(i)) - R_{j^*}$, which immediately gives $\sum_{i = 1}^n U_i(A_{j}, P) = \sum_{i = 1}^n U_i(A_{j^*}, P)$. If Equation \eqref{eq:contra} does not hold, then there exists some player $i'$ such that $U_{i'}(A_{j^*}, P) > U_{i'}(A_{j}, P)$. This implies that
\[
    \sum_{i \neq i'} U_i(A_{j}(i)) - \sum_{i \neq i} U_i(A_{j^*}(i)) > 0.
\]
By the pigeonhole principle, at least one player (which is not $i'$) must therefore have strictly higher utility in $j$ than $j^*$. This is a contradiction with the fact that $j^*$ is a consensus apartment. 

\subsection{Proof of Lemma \ref{obs:ref}}\label{app:trading_equivalence}

\begin{lemma}[``If" direction of Lemma \ref{obs:ref}]\label{lemma:if_direction}
        A partial solution $(A,P)$ is reachable by negotiations from an individual envy-free starting partial solution $(A,Q)$ if $\sum_{j = 1}^m P(A_j(i)) = \sum_{j = 1}^m Q(A_j(i))$ for all players $i$.
\end{lemma}

\begin{proof}
    As in the problem statement, let $(A,Q)$ be an individually envy-free partial solution and $(A,P)$ be a partial solution such that $ \sum_{j = 1}^m P(A_j(i)) = \sum_{j = 1}^m Q(A_j(i))$ for all players $i$. We will provide a constructive algorithm for finding a series of negotiations that transform the partial solution $(A,Q)$ into the partial solution $(A,P)$. Specifically, we will find a sequence of partial solutions $(A,Q), (A,Q_1), (A,Q_2),...,(A,Q_{m-1})$ such that the partial solution $(A,Q_1)$ is reachable by negotiations from $(A,Q)$, the partial solution $(A,Q_j)$ is reachable from $(A,Q_{j-1})$ for all $j \le m-1$, and $(A,Q_{m-1}) = (A,P)$.

    For any price matrix $P'$, define the $n$ by $m$ matrix $\Delta_P(P')$ as
    \[
        \Delta_P(P')_{ij} =  P'(A_i(j)) - P(A_i(j)).
    \]
    Because $P$ and $P'$ are both valid price matrices, the sum of the prices within any apartment is equal to the rent of that apartment for both $P$ and $P'$, which implies that for all $j$,
    \begin{equation}\label{eq:within_apt}
        \sum_{i=1}^n \Delta_P(P')_{ij} = 0.
    \end{equation}
    If we apply Equation \eqref{eq:within_apt} to apartment $1$ and price matrix $Q$, we observe that $ \sum_{i=1}^n \Delta_P(Q)_{i1} = 0$. This implies that 
    \begin{equation}
        \sum_{i=1, \Delta_P(Q)_{i1} > 0}^n \Delta_P(Q)_{i1} = -\sum_{i=1, \Delta_P(Q)_{i1} < 0}^n \Delta_P(Q)_{i1}.
    \end{equation}
    Therefore, there must exist a sequence of negotiations that occur only between apartment $1$ and apartment $m$ that result in a partial solution $(A,Q_1)$ such that $\Delta_P(Q_1)_{i1} = 0$ for all $i$. By definition of a valid negotiation, it must still be true that $ \sum_{j = 1}^m P(A_j(i)) = \sum_{j = 1}^m Q_1(A_j(i))$ after each negotiation. Furthermore, by construction we have $Q_1(A_i(1)) = P(A_i(1))$ for all players $i$.

    The same logic allows us to conclude that there exists a sequence of negotiations between apartment $2$ and apartment $m$ that transform the partial solution $(A,Q_1)$ into the partial solution $(A,Q_2)$ such that  $ \sum_{j = 1}^m P(A_j(i)) = \sum_{j = 1}^m Q_2(A_j(i))$ and $\Delta_P(Q_2)_{i2} = 0$ for all players $i$. Note that by this construction, in the partial solution $(A,Q_2)$, we have  $Q_2(A_i(j)) = P(A_i(j))$ for all players $i$ and $j \le 2$. 
    
    We can recursively continue this negotiating process between apartments $j$ and $m$ up until $j = m-1$. At that point, we will have a partial solution $(A,Q_{m-1})$ reachable by negotiations from $(A,Q)$ such that the following two equations hold:
    \begin{equation}\label{eq:mminusone}
        \sum_{j = 1}^m P(A_j(i)) = \sum_{j = 1}^m Q_{m-1}(A_j(i)) \quad \forall i
    \end{equation}
    \begin{equation}\label{eq:mminusoneb}
        Q_{m-1}(A_i(j)) = P(A_i(j)) \quad \forall i, \forall j \le m-1
    \end{equation}
    Note that plugging Equation \eqref{eq:mminusoneb} into Equation \eqref{eq:mminusone} gives the following result.
    \[
        Q_{m-1}(A_i(m)) = P(A_i(m)) \quad \forall i
    \]
    Therefore, we have shown that $Q_{m-1} = P$. Since $(A,Q_{m-1})$ was reachable by negotiations from $(A,Q)$, this implies that $(A,P)$ is reachable by negotiations from $(A,Q)$, which is the result we wanted to show.
\end{proof}

\subsection{Proof of Property \ref{prop:pareto}}
\begin{proof}
    We know that an envy-free solution in a single apartment $j$ must include a welfare-maximizing assignment for $j$~\cite{GMPZ17}. A negotiated envy-free solution $(A, P, j^*)$ thus must include a welfare-maximizing assignment in every apartment, as the assignment $A$ satisfies for some price matrix that $(A,Q)$ is individually envy-free. From Lemma \ref{lemma:highest_max_welfare}, we know that consensus can only occur in an apartment with the highest welfare across all apartments. Because $A$ must contain only welfare-maximizing assignments and $j^*$ must be the apartment with the highest welfare over all apartments, assignment $A_{j^*}$ must achieve the maximum welfare among all possible assignments and apartments. This implies that $(A, P, j^*)$ is Pareto optimal.
\end{proof}

\subsection{Proof of Property \ref{prop:ir}}
\begin{proof}
    Let $(A, P, j^*)$ be a solution that satisfies negotiated envy-freeness, and consider the bundle of rooms $\{A_j(i)\}_{j=1}^m$ which is assigned to player $i$ under assignment $A$. We say that player $i$'s total utility for the bundle $\{A_j(i')\}_{j=1}^m$ of player $i'$ is equal to the sum of the utilities that player $i$ has for each room in $\{A_j(i')\}_{j=1}^m$. In an individually envy-free partial solution $(A,Q)$, every player $i$ has weakly higher utility for her own bundle $\{A_j(i)\}_{j=1}^m$ than any other player's bundle $\{A_j(i')\}_{j=1}^m$. By definition of negotiated envy-freeness, every player in the partial solution $(A,P)$ has the same total utility for every bundle as in some individually envy-free partial solution $(A,Q)$. This implies that in partial solution $(A,P)$, player $i$ has weakly higher total utility for the bundle $\{A_j(i)\}_{j=1}^m$ than the bundle $\{A_j(i')\}_{j=1}^m$. The sum of each player's utility over all players' bundles is $0$ by the assumption that each player's total value for all rooms is equal to the total rent. This implies that player $i$'s average utility over all bundles is $0$, which in turn means that her value for $\{A_j(i)\}_{j=1}^m$, her favorite bundle, must be at least $0$ as well. Finally, we know that $j^*$ is a consensus apartment, which means that player $i$'s favorite room within her bundle is in apartment $j^*$. By the same averaging argument, we conclude that player $i$ has non-negative utility for her assigned room in apartment $j^*$, or equivalently that $U_i(A_{j^*}(i)) \geq 0$.
\end{proof}

\subsection{Proof of Property \ref{prop:single}}
\begin{proof}
    In the single-apartment setting, consensus is always satisfied and no negotiations are possible. Therefore, any envy-free solution within the one apartment satisfies consensus and hence satisfies negotiated envy-freeness. For the other direction, any negotiated envy-free solution must be an individually envy-free solution as no negotiations are possible, and individual envy-freeness implies envy-freeness in the lone apartment.
\end{proof}
\subsection{Proof of Theorem \ref{thm:ref_existence}}\label{app:proof_of_ref_existence}
\begin{proof}
    Our proof is by construction of such a solution $(A,P^*,j^*)$ for any multi-apartment rent division instance. Note that throughout the proof, we will be referring to the same assignment $A$, but different price matrices $Q, P,$ and $P^*$. 
    
    We begin with a partial solution $(A, Q)$ which is individually envy-free. Such a partial solution $(A,Q)$ always exists because an envy-free solution always exists in the one apartment setting. As in the proof of individual rationality (Property~\ref{prop:ir}), we define player $i$'s bundle as $\{A_j(i)\}_{j=1}^m$, i.e., the set of rooms assigned to player $i$ across all apartments. The price of each bundle will be the sum of prices of rooms in that bundle; for notational convenience, we will overload the operator $P$ and let $P(\{A_j(i)\}_{j=1}^m)=\sum_{j = 1}^m P(A_j(i))$. Our goal will be to find a price matrix $P^*$ such that $Q(\{A_j(i)\}_{j=1}^m) = P^*(\{A_j(i)\}_{j=1}^m)$ and there exists a consensus apartment $j^*$ in the partial solution $(A, P^*)$. This would then give a negotiated envy-free solution $(A,P^*, j^*)$.

    Let $\mathcal{P}_j$ be the set of all price vectors in apartment $j$ that add up to the total rent. First, consider the partial solution $(A, P)$ where $P$ satisfies for all apartments $j$,
    \[
        P_j \in \arg\max_{P'_j \in \mathcal{P}_j} \min_i V_i(A_j(i)) - P'_j(A_j(i)).
    \]
    This partial solution guarantees that for each apartment $j$, the utilities of all players for their assigned rooms in $j$ will be equal. Therefore, the apartment $j^*$ with the highest value of $\min_i V_i(A_{j^*}(i)) - P_{j^*}(A_{j^*}(i))$ will be a consensus apartment. 
    
    Unfortunately, it is not necessarily true that $Q(\{A_j(i)\}_{j=1}^m) = P(\{A_j(i)\}_{j=1}^m)$ for all $i$. Therefore, we will construct another price assignment $P^*$ from $P$ that has this desired property. For every player $i$, define $X_i = Q(\{A_j(i)\}_{j=1}^m) - P(\{A_j(i)\}_{j=1}^m)$.
    
    We define the new price matrix $P^*$ such that $P^*(A_j(i)) = P_i + \frac{X_i}{m}$ for every $i,j$. In other words, we increase each of player $i$'s prices in $P$ by $\frac{X_i}{m}$. By this construction, we have that $P^*(A_j(i)) - P(A_j(i)) = P^*(A_{j'}(i)) - P(A_{j'}(i)) \quad \forall j, j' \in [m], \forall i \in [n]$.
    
    This implies that each player's preferences over rooms are the same in $P$ and $P^*$, and therefore the consensus apartment $j^*$ for $(A, P)$ is also a consensus apartment for $(A, P^*)$. Furthermore, we have for every $i$ that
    \begin{align*}
        P^*(\{A_j(i)\}_{j=1}^m) &= \sum_{j = 1}^m \left( P(A_j(i)) + \frac{X_i}{m} \right) \\
        &= X_i + P(\{A_j(i)\}_{j=1}^m) \\
        &= Q(\{A_j(i)\}_{j=1}^m)
    \end{align*}
    as desired. 
    
    Our final step is to show that $P^*$ is a valid price matrix in the sense that the prices in each apartment $j$ add up to the rent $R_j$. Because we started from a valid price matrix $P$, for each apartment $j$ we have that
    \[
        \sum_{i = 1}^n P^*(A_j(i)) = \sum_{i = 1}^n \left(P(A_j(i)) + \frac{X_i}{m} \right) = R_j + \frac{1}{m} \sum_{i =1}^n X_i.
    \]
    Therefore, it suffices to prove that $\sum_{i =1}^n X_i = 0$. We know that both $P$ and $Q$ are valid price matrices, and so we can conclude that
    \begin{align*}
        \sum_{i =1}^n X_i &= \sum_{i =1}^n Q(\{A_j(i)\}_{j=1}^m) - \sum_{i =1}^n P(\{A_j(i)\}_{j=1}^m)  \\
        &= \sum_{i =1}^n \sum_{j=1}^m Q(A_j(i)) -  \sum_{i =1}^n  \sum_{j=1}^m P(A_j(i)) \\
        &= \sum_{j=1}^m \sum_{i =1}^n Q(A_j(i)) - \sum_{j=1}^m \sum_{i =1}^n   P(A_j(i)) \\
        &= \sum_{j =1}^m R_j -  \sum_{j =1}^m R_j \\
        &= 0.
    \end{align*}
    Thus, $P^*$ is a valid price matrix. In summary, we have shown that $(A,P^*)$ is a valid partial solution, that $(A,P^*)$ satisfies consensus, and that $Q(\{A_j(i)\}_{j=1}^m) = P^*(\{A_j(i)\}_{j=1}^m)$. Therefore, letting $j^*$ be a consensus apartment for $(A,P^*)$, we have that $(A,P^*,j^*)$ is a solution that satisfies negotiated envy-freeness.
\end{proof}

\subsection{Proof of Lemma \ref{lemma:multi_second_welfare}}\label{app:multi_second_welfare}
\begin{proof}
    Define $P'$ as follows:
    \[
        P'(A'_j(i)) = V_i(A'_j(i)) - V_i(A_j(i)) + P(A_j(i)).
    \]
    We will prove the following three claims regarding $P'$, which will together prove the theorem statement. First, we show that $P'$ is a valid price matrix, i.e. for all apartments $j$, $\sum_i P'(A'_j(i)) = R_j$. Next, we show that for all players $i$ and apartments $j$, $U_i(A_j, P) = U_i(A_j', P')$. Finally, we show that $(A', P', j^*)$ satisfies negotiated envy-freeness. 

    We first show that $P'$ is a valid price matrix. For every $j$, we have
    \begin{align*}
        \sum_{i = 1}^n P'(A'_j(i)) &= \sum_{i = 1}^n V_i(A'_j(i)) - \sum_{i = 1}^n V_i(A_j(i)) + \sum_{i = 1}^n P(A_j(i)) \\
        &= \sum_{i = 1}^n P(A_j(i)) \\
        &= R_j 
    \end{align*}
    where we used that $A, A'$ are  welfare-maximizing assignments and that $P$ is a valid price matrix. 

    Next, we show that $U_i(A_j, P) = U_i(A_j', P')$ for all $i$ and $j$. By our choice of $P'$, we have
    \[
        U_i(A_j, P) = V_i(A_j(i)) - P(A_j(i)) = V_i(A_j'(i)) - P'(A_j'(i)) = U_i(A_j', P')
    \]
    as desired. Note that this implies that $(A', P', j^*)$ satisfies consensus because $(A, P, j^*)$ satisfies consensus.

    We next show that $(A', P', j^*)$ satisfies negotiated envy-freeness. We already showed that $(A',P',j^*)$ satisfies consensus. As $(A, P, j)$ satisfies negotiated envy-freeness, we know that there must be some price matrix $Q$ such that $(A, Q)$ is individually envy-free and $\sum_j Q(A_j(i)) = \sum_j P(A_j(i))$. Through repeated applications of Lemma \ref{lemma:second_welfare}, we observe that $(A', Q)$ is individually envy-free as well. In order to show that $(A', P',j^*)$ satisfies negotiated envy-freeness, it therefore suffices to show that for all players $i$, $\sum_{j = 1}^m P'(A'_j(i)) = \sum_{j = 1}^m Q(A'_j(i))$. 
    
    Because $A', A$ are welfare maximizing assignments and $(A,Q)$ is individually envy-free, Lemma \ref{lemma:second_welfare} implies that for any $j$, $V_i(A'_j(i)) - Q(A'_j(i)) = V_i(A_j(i)) - Q(A_j(i))$ (which is used in the final line below). Furthermore, recall that we chose $Q$ such that $\sum_j Q(A_j(i)) = \sum_j P(A_j(i))$ (which is used in the second line below). By construction of $P'$ and $Q$, we can therefore conclude that
    \begin{align*}
        \sum_{j = 1}^m P'(A'_j(i)) &= \sum_{j = 1}^m \big(V_i(A'_j(i)) - V_i(A_j(i)) + P(A_j(i))\big) \\
        &= \sum_{j = 1}^m \big(V_i(A'_j(i)) - V_i(A_j(i)) + Q(A_j(i))\big) \\
        &= \sum_{j = 1}^m Q(A'_j(i)).
    \end{align*}
    Therefore, we have shown that $(A',P',j^*)$ satisfies negotiated envy-freeness.
\end{proof}

\subsection{Proof of Theorem \ref{thm:max_obj}}\label{app:thm:max_obj}
\begin{proof}

In order to prove Theorem \ref{thm:max_obj}, we recall the following key result from the single apartment setting which implies that we can start from any welfare-maximizing assignment when optimizing over envy-free solutions. This is necessary because the number of welfare-maximizing assignments could potentially be exponential in $n$. For the reader's convenience, we will state and give a proof sketch of this result below.

\begin{lemma}[2nd Welfare Theorem \cite{GMPZ17,mas1995microeconomic}]\label{lemma:second_welfare}
    For a single apartment $j$, if $(A_j, P_j)$ is an envy-free solution and $A_j'$ is a welfare-maximizing assignment, then for all $i \in [n]$,
    \[
        V_i(A_j(i)) - P(A_j(i)) = V_i(A_j') - P(A_j'(i)).
    \]
    This further implies that $(A_j', P_j)$ is an envy-free solution.
\end{lemma}
\begin{proof}
    It holds that $U_i(A_j, P_j) \ge U_i(A_j', P_j)$ for all players $i$ because $(A_j,P_j)$ is envy-free. We also know that because $A_j'$ is a welfare-maximizing assignment, $\sum_{i=1}^n U_i(A_j,P_j) \le \sum_{i=1}^n U_i(A_j', P_j)$. Therefore, no player $i$ can have $U_i(A_j, P_j) > U_i(A_j', P_j)$, as this would imply that there is some player $i'$ such that  $U_{i'}(A_j, P_j) < U_{i'}(A_j', P_j)$. Therefore, we can conclude that $U_i(A_j, P_j) = U_i(A_j', P_j)$  for all players $i$. Since $(A_j, P_j)$ was an envy-free solution, these utilities being equal implies that $(A_j', P_j)$ must be an envy-free solution as well.
\end{proof}

In Lemma \ref{lemma:second_welfare}, if $(A_j,P_j)$ is envy-free, then a different maximum-welfare assignment $A'_j$ immediately yields an envy-free solution $(A'_j, P_j)$ where all players have the same utilities. However, in the multi-apartment setting, if $(A,P,j^*)$ satisfies negotiated envy-freeness, we cannot simply use a different welfare-maximizing assignment $A'$ with the same price matrix to find a negotiated envy-free solution $(A', P, j^*)$. Surprisingly, however, we are able to construct another price matrix $P'$ such that $(A', P', j^*)$ and $(A,P,j^*)$ have the same utilities for every player in their assigned room in every apartment. This is exactly the result of Lemma \ref{lemma:multi_second_welfare}.

    For a rent division instance $V$, we are able to find some maximum welfare assignment $A$ in polynomial time using maximum weight bipartite matching separately in each apartment. Fix any apartment $j'$. We can find the optimal price matrix $P$ for assignment $A$ subject to negotiated envy-freeness with $j'$ as the consensus apartment in polynomial time using the following linear program. Note that this linear program will have no solution if there does not exist a price matrix $P$ such that $j'$ is a consensus apartment for $(A,P)$.

    \begin{align*}\label{lp:1}
        \max_{P,Q} \quad & Z \\
        \text{s.t.} \quad & Z \leq f_q(U_1(A_{j'},P),...,U_n(A_{j'},P)) & \forall \: q & \qquad  \text{[min over $f_q$]}\\
        & V_i(A_{j'}(i)) - P(A_{j'}(i)) \geq V_i(A_{j}(i)) - P(A_{j}(i)) & \forall \: i, j & \qquad \text{[consensus in $j'$]}\\
        & V_i(A_j(i)) - Q(A_j(i)) \geq V_i(A_j(i')) - Q(A_j(i')) & \forall \: i, i', j & \qquad \text{[$Q$ is individually EF]} \\
        & \sum_j P(A_j(i)) = \sum_j Q(A_j(i))& \forall \: i & \qquad \text{[negotiated EF]} \\
        & \sum_i P(A_j(i)) = R_j & \forall \: j & \qquad \text{[payments equal rent]} \\
        & & \tag{LP1}
    \end{align*}
     In order to find the optimal value of the objective function over all negotiated envy-free solutions with assignment $A$, we also need to optimize the objective over all choices of apartment $j'$. We could simply run this linear program for all $j' \in [m]$ and choose the $j'$ that achieves the highest objective value. However, we will instead choose $j'$ in such a way that we only need to run this linear program once. From Lemma \ref{lemma:highest_max_welfare}, we know that if $(A,P)$ satisfies consensus, then apartment $j_A^*$ is a consensus apartment for $(A,P)$ if and only if $\sum_{i = 1}^n V_i(A_{j_A^*}(i)) - R_{j_A^*} = \max_j \sum_{i = 1}^n V_i(A_{j}(i)) - R_{j}$. Furthermore, Lemma \ref{lemma:highest_max_welfare} implies that for any two consensus apartments $j_1,j_2$ in $(A,P)$, every player must have the same utility in $j_1$ as they do in $j_2$. Therefore, in the linear program we can set $j'$ equal to any $j_A^*$ which satisfies $\sum_{i = 1}^n V_i(A_{j_A^*}(i)) - R_{j_A^*}\geq \max_j \sum_{i = 1}^n V_i(A_{j}(i)) - R_{j}$, and this will be equivalent to optimizing over all choices of $j'$.
    
    Finally, we show that the optimal solution assuming an arbitrary maximum welfare assignment $A$ is in fact a globally optimal solution over all negotiated envy-free solutions. Let the globally optimal solution over all negotiated envy-free solutions be $(A^*, P^*, j^*)$. By Lemma \ref{lemma:multi_second_welfare}, we know there exists a $P'$ such that $U_i(A^*_j, P^*) = U_i(A_j, P') \: \forall i, j$. Note that this implies that $j^*$ is a consensus apartment for $(A, P')$. Then, because $(A,P,j_A^*)$ was the optimal solution for all negotiated envy-free solutions with assignment $A$, we know that
    \begin{align*}
        &f_q(U_1(A_{j_A^*},P),...,U_n(A_{j_A^*},P)) \\
        &\geq f_q(U_1(A_{j^*},P'),...,U_n(A_{j^*},P')) \\
        &= f_q(U_1(A^*_{j^*},P^*),...,U_n(A^*_{j^*},P^*))
    \end{align*}
    Therefore, it suffices to run the linear program with any arbitrary welfare-maximizing assignment $A$. The full algorithm is shown in the body in Algorithm \ref{algo:REF_consensus_opt}.
\end{proof}

\section{Distribution Envy-Freeness}\label{app:distributional}

 Up until this point, we have required that the solution be of the form $(A,P,j)$, where $j$ is the chosen apartment. In this section, we will consider the natural generalization of a distributional solution of the form $(A,P,D)$, where $D$ is a distribution over all $m$ apartments. Generalizing the notion of envy-freeness in the one apartment case, we want to choose a distribution $D$ such that no player is envious of any other player in expectation. Formally, we want to choose a distribution $\mathcal{D}$ such that
\begin{equation}\label{eq:distr_EF_pref}
                \sum_{j=1}^m D[j] \cdot \left(V_i(A_j(i)) - P(A_j(i))\right) \ge \max_{i'} \sum_{j=1}^m D[j] \cdot \left(V_i(A_j(i')) - P(A_j(i'))\right) \quad \forall \: i.
\end{equation}
This requirement alone is actually easy to satisfy, as we could simply choose any individually envy-free partial solution $(A,Q)$. Then for any distribution $D$, no player will be envious in expectation under distributional solution $(A,P,D)$. In fact, equation \eqref{eq:distr_EF_pref} is too easy to satisfy because it does not have any consensus-like requirement that players do not prefer other distributions over apartments. Therefore, we also want to require that every player $i$ prefers the distribution $D$ to any other distribution $D'$ under partial solution $(A,P)$. Formally, we will require that
\begin{equation}\label{eq:dist_cons}
                \sum_{j=1}^m D[j] \cdot \left(V_i(A_j(i)) - P(A_j(i))\right) \ge \max_{D'} \sum_{j=1}^m D'[j]  \left(V_i(A_j(i)) - P(A_j(i))\right) \quad \forall \: i.
\end{equation}
We will define a distributional solution $(A,P,D)$ as\textit{ distribution envy-free (DEF) }if Equations \eqref{eq:distr_EF_pref} and \eqref{eq:dist_cons} are both satisfied. Note that if the solution  $(A,P,j^*)$ is universal envy-free, then the distributional solution $(A,P,D)$, where $D$ chooses apartment $j^*$ with probability $1$, must be distribution envy-free. Furthermore, distribution envy-free and consensus have the following relationship.

\begin{lemma}
    If the distributional solution $(A,P,D)$ is distribution envy-free, then any solution of the form $(A,P,j^*)$ for any $j^*$ with non-$0$ weight under distribution $D$ will satisfy consensus.
\end{lemma}
\begin{proof}
    Proof by contradiction. Suppose $(A,P,D)$ is a distributional solution that is distribution envy-free, and there exists a $j^*$ that has positive probability under distribution $D$, but $(A,P,j^*)$ does not satisfy consensus. Because $(A,P,j^*)$ does not satisfy consensus, there must exist some player $i$ and apartment $j' \ne j^*$ such that $U_i(A_{j'}, P) \ge U_i(A_{j}, P)$ for all apartments $j$ and $U_i(A_{j'}, P) > U_i(A_{j^*}, P)$. Consider the distributional solution $(A,P,D')$, where $D'$ with probability $1$ chooses apartment $j'$. Then
\begin{align*}
                \sum_{j=1}^m D[j] \cdot \left(V_i(A_j(i)) - P(A_j(i))\right) <& \sum_{j=1}^m D[j] \cdot \left(V_i(A_{j'}(i)) - P(A_{j'}(i))\right) \\
                &= V_i(A_{j'}(i)) - P(A_{j'}(i)) \\
                &= \sum_{j=1}^m D'[j]  \left(V_i(A_j(i)) - P(A_j(i))\right).
\end{align*}
This violates Equation \eqref{eq:dist_cons} and is a contradiction to the assumption that $(A,P,D)$ is distribution envy-free.
\end{proof}

Distribution envy-freeness can be seen as a compromise between consensus and universal envy-freeness. To better understand the relationship between universal envy-freeness and distribution envy-freeness, recall  Example \ref{example:motivating}. In this example, the two players have symmetric utility functions, but each prefers a different apartment. Despite this symmetry, there does not exist a universal envy-free solution for this example. However, there does exist a distributional solution that is distribution envy-free. One such solution is to assign rents equal to the assigned player's values in every room, and then take the distribution $D$ to be uniform over the two apartments. Under this solution, both players have utility of $0$ for both of their assigned rooms. Furthermore, both players in expectation have utility $0$ for the other player's assigned rooms under distribution $D$. Therefore Equations \eqref{eq:distr_EF_pref} and  \eqref{eq:dist_cons} are both satisfied and this solution is distribution envy-free. Furthermore, this solution is intuitively fair, in the sense that the problem is perfectly symmetric and the solution also treats both players perfectly symmetrically. 

 Unfortunately, a distributional solution that is distribution envy-free does not always exist, as can be seen by Example \ref{example:EEF}. However, Equations \eqref{eq:dist_cons} and \eqref{eq:distr_EF_pref} are both linear constraints on the prices $P$. Therefore, as for universal envy-freeness, it is possible check for existence of a distributional solution satisfying distribution envy-freeness in polynomial time.

\section{Apartment Monotonicity of Negotiated Envy-Freeness}\label{app:monotonicity}

In this section we explore whether adding more apartments always implies higher objective values for the optimal solution under negotiated envy-freeness. Specifically, suppose that we start with $m$ apartments, and consider the optimal solution under negotiated envy-freeness for objective functions $f_1,...,f_t$. We would like to know whether the optimal solution after adding an $m + 1$th apartment always increases in value, and define this guarantee formally as apartment monotonicity.

\begin{definition}
    Let $f_1,...,f_t: \mathbb{R}^{n \times m} \to \mathbb{R}$ be linear functions, where $t$ is polynomial in $n$ and $m$. Let $(A, P, j^*)$ be the solution which maximizes the minimum of $f_q(U_1(A_{j^*},P),...,U_n(A_{j^*},P))$ over all $q \in [t]$ subject to negotiated envy-freeness given instance $V \in \mathbb{M}_{n \times m \times n}(\mathbb{R}^+)$. Let $V' \in \mathbb{M}_{n \times (m + 1) \times n}(\mathbb{R}^+)$ be a valuation matrix such that $V' = \big[V \: \: V_{m+1}\big]$. Finally, let $(A', P', j')$ be the solution which maximizes the minimum of $f_q(U_1(A'_{j'},P'),...,U_n(A'_{j'},P'))$ over all $q \in [t]$ subject to negotiated envy-freeness given instance $V'$. Then we say that $f_1,...,f_t$ satisfies \textbf{apartment monotonicity} if $\: \forall \: V'$,
    \[
        f_q(U_1(A'_{j'},P'),...,U_n(A'_{j'},P')) \geq f_q(U_1(A_{j^*},P),...,U_n(A_{j^*},P)).
    \]
\end{definition}

Note that as we require that players have the same total value for all apartments under consideration, and the players do not change their values for the first $m$ apartments, it must be the case that every player has the same total value for the rooms in apartment $m + 1$. We show below that the maximin objective function does not satisfy apartment monotonicity.

\begin{lemma}
    Let the maximin objective function be the linear functions $f_1,...,f_n$ such that 
    \[
        f_i(U_1(A_{j^*},P),...,U_n(A_{j^*},P) = U_i(A_j^*, P)
    \] 
    for all $i \in [n]$. Then the maximin objective function does not satisfy apartment monotonicity. 
\end{lemma}
\begin{proof}
Consider below the following two apartments, both with total rent $300$. A maximum welfare assignment in each apartment is in bold.

    \parbox{.45\linewidth}{
        \centering
        \begin{tabular}{c | c c c}
        & $r_{11}$ & $r_{12}$ & $r_{13}$ \\
        \hline
        1 & \textbf{150} &   150 &   0 \\
        2 & 0 &   \textbf{150} &   150 \\
        3 & 75 & 75 & \textbf{150} \\
        \end{tabular}
        }
        \hfill
        \parbox{.45\linewidth}{
        \centering
        \begin{tabular}{c | c c c}
        & $r_{21}$ & $r_{22}$ & $r_{23}$ \\
        \hline
        1 & 100 & 100 & \textbf{100} \\
        2 & 300 &   \textbf{0} &   0 \\
        3 & \textbf{300} &   0 &   0 \\
        \end{tabular}
    }
    
Let the partial solution $(A, P)$ be individually envy-free, which means that $A$ must be a maximum welfare assignment. Then in apartment $1$, we must have 
\[
    P(r_{11}) \leq P(r_{12}) \leq P(r_{13}) \leq P(r_{11}) + 75.
\]
Violating the first inequality would result in player 1 envying player 2, violating the second would result in player 2 envying player 3, and violating the third would result in player 3 envying player 1. Therefore, the maximum rent that player $1$ can pay in an envy-free solution for apartment $1$ is $300/3 = 100$, which implies that the minimum utility player 1 can have in apartment $1$ is $50$. The minimum rent player $1$ can pay in an envy-free solution for apartment $1$ is $50$, which implies that the maximum utility player $1$ can have in apartment $1$ is $100$. In apartment $2$, there is only one envy-free solution; we must have $P(r_{21}) = 300, P(r_{22}) = 0,$ and $P(r_{23}) = 0$. In this solution, player $1$ has utility $100$ in apartment $2$, while players $2$ and $3$ each have utility $0$. 

Let $V_j \in \mathbb{M}_{3  \times 3}$ be the valuation matrix of player values in apartment $j$. Further let $V' = [V_1 V_2]$. Let $(A_{V_1}^*, P_{V_1}^*, 1)$ be the negotiated envy-free solution given instance $V_1$ which has the highest maximin value. Note that because negotiated envy-freeness reduces to envy-freeness in the single apartment setting, this is equivalent to the envy-free solution in apartment $1$ which has the highest maximin value. We can therefore observe that $P_{V_1}^*$ is the symmetric price matrix where $P_{V_1}^*(r_{1k}) = 100$ for $k \in [3]$, which results in each player having utility $50$. Note that the maximin value cannot be higher because the maximum utility of player $3$ in apartment $1$ is $50$. Furthermore, when player $3$'s utility in apartment $1$ is $50$, it must be the case that the utilities of player $1$ and $2$ in apartment $1$ are $50$ as well. 

Let $(A_{V'}^*, P_{V'}^*, j^*)$ be the negotiated envy-free solution given $V'$ which has the highest maximin value, and let $F(A_{V'}^*, P_{V'}^*, j^*)$ be its maximin value. We would like to show that $F(A_{V'}^*, P_{V'}^*, j^*)$ is strictly lower than $50$. Note that
\[
    F(A_{V'}^*, P_{V'}^*, j^*) = \max_{j \in \{1, 2\}} F(A_{V'}^*, P_{V'}^*, j).
\]
Therefore, it suffices to show that the objective value of $(A_{V'}^*, P_{V'}^*, j)$ is less than $50$ both when apartment $1$ is the consensus apartment, and when apartment $2$ is the consensus apartment. 

First, assume apartment $1$ is the consensus apartment. Observe that $\sum_{j = 1}^2 P_{V'}(A_j(1)) \leq 100$, i.e. that player $1$'s rent burden over both apartments is at most $400$. Because apartment $1$ is the consensus apartment, player $1$ must have utility at least
\[
    \frac{V'_1(A_1(i)) + V'_1(A_2(i)) - 400}{2} = 75.
\]
However, we established earlier that the only way for player $3$'s utility in apartment $1$ to be $\geq 50$ is for player $1$'s utility in apartment $1$ to equal $50$. Therefore, $F(A_{V'}^*, P_{V'}^*, 1) < 50$. Now, assume apartment $2$ is the consensus apartment. Observe that $\left(\sum_{i = 1}^3 V'_i(A_2(i))\right) - R_2 = 100$.  This implies that the total utility of the players in apartment $2$ for any solution is at most $100$, and the maximin value can therefore be at most $100/3$. Therefore, $F(A_{V'}^*, P_{V'}^*, 2) < 50$. This shows that $F(A_{V'}^*, P_{V'}^*, j^*) < 50$, as desired.
\end{proof}

It may seem surprising that adding more apartments can lead to a decrease in the utility of the least happy player. However, this result makes sense when we consider that adding more apartments may lead to a wider spread of player preferences over apartments. In particular, adding a new apartment may make consensus more difficult to achieve, as one or more players may find that the new apartment is now their favorite apartment. Recall that in the negotiated envy-free setting, compromising in order to reach consensus takes the form of a player paying more rent in a favorable apartment in order to pay less rent in a less favorable apartment. As more apartments are added, players favoring the consensus apartment may need to pay more in order to convince other players to join them, which in turn may decrease the maximin value.

Note that it is also possible that adding a new apartment increases the maximin value. For example, suppose that apartment $2$ instead consisted of each player preferring a different room, with their value for that room being equal to the total rent. Then the consensus apartment for instance $V'$ would be apartment $2$, and the maximin value would be $200$. Hence, the effect of adding a new apartment on the maximin value depends on the structure of the valuation functions of the players for the new apartment.

\section{Core}\label{app:core}

Throughout our work, we have focused on the setting where there are $n$ players who are considering only $n$-bedroom apartments. An extension would be to consider when apartments of size $<n$ are available as well, and players are willing to split into groups to rent multiple apartments. In this setting, each apartment $j$ has a size $s(j)$. Each player $i$ still has a non-negative value $V_i(r_{jk})$ for any room $r_{jk}$ in apartment $j$. We will define a solution to this problem as a tuple $(A,P,J)$ where $J \subseteq [m]$ such that $\sum_{ j \in J} s(j) = n$, $A$ is a mapping from players to rooms in the apartments in $J$, and $P$ is a price map for all rooms in apartments in $J$ such that for any $j \in J$, $\sum_{r_{jk} \in j} P(r_{jk}) = R_j$.  Intuitively, this last requirement means that the subset of players renting any apartment is responsible for paying exactly the rent of that apartment (without subsidizing or subsidization from players who have been assigned to other apartments). Note that unlike in the rest of the paper, $A$ is now a one-to-one map from players to rooms in the apartments in $J$.

We are interested in whether we can always find a stable solution in which no set of players can achieve higher total utility by deviating to a set of unoccupied apartments. When there are a finite number of copies of each apartment type, we can easily see that the answer is no. Specifically, consider the following example with two players. There are two apartments both with rent $0$. The first apartment has two rooms and both players have utility $-1$ for both rooms. The second apartment has one room and both players have utility $1$ for that room. Then the only valid solution is to assign both players to the two-bedroom apartment, but both players will want to deviate to the one-bedroom apartment.

When there are infinite copies of each type of apartment, it is less straightforward to determine whether such a solution always exists. With infinite copies of apartments, we can treat this problem as a cooperative game with transferable utility. A natural question in cooperative game theory is whether the core is non-empty. In our setting, this boils down to asking whether we can always find a solution $(A, P, J)$ such that no group of players $S$ would want to deviate to a valid set of apartments $J' \subseteq [m]$. Formally, let $S \subseteq [n]$ be a coalition of $|S|$ players and let $\mathcal{S}$ be the set of all valid solutions $(A',P',J')$ for just these $S$ players. Therefore $\sum_{j \in J'} s(j) = |S|$ and $A'$ only assigns players in the coalition $S$ to rooms in the apartments in $J'$. We can then define the value function $v$ for the coalition $S$ as 
\[
    v(S) = \max_{(A',P',J') \in \mathcal{S}} \sum_{i \in S} V_i(A'(i)) - P'(A'(i)).
\]
In this definition, $v(S)$ can be thought of as the utility that the coalition $S$ can get if they deviate. Define $\alpha_i$ as the utility of player $i$ under solution $(A, P, J)$, i.e. $\alpha_i = V_i(A(i)) - P(A(i))$. Then the core is the set of all solutions $(A, P, J)$ such that  
\begin{equation}\label{eq:core_pre}
    \sum_{i \in S} \alpha_i \geq  v(S) \quad \forall S \subseteq [n].
\end{equation}
Note that this is equivalent by construction a solution $(A,P,J)$ such that
\begin{equation}\label{eq:core}
    \sum_{i \in S}  V_i(A(i)) - P(A(i)) \ge \max_{(A',P',J') \in \mathcal{S}} \sum_{i \in S} V_i(A'(i)) - P'(A'(i)) \quad \forall S \subseteq [n]
\end{equation}
We show below that the core can be empty for the multi-apartment rent divsion problem.

\begin{lemma}
    There exists instances of the multi-apartment rent division problem for which the core as described above is empty.
\end{lemma}
\begin{proof}
    \begin{example}\label{example:core}
        Consider the following example where there are three types of apartments and the rent is $0$ in every apartment.
    
        \parbox{.25\linewidth}{
            \centering
            \begin{tabular}{c | c c c}
            & $r_{11}$ & $r_{12}$ & $r_{13}$ \\
            \hline
            1 & 340 & 340 &  340\\
            2 & -20 & -20  & -20\\
            3 & -20 & -20  & -20\\
            \end{tabular}
        }
        \parbox{.25\linewidth}{
            \centering
            \begin{tabular}{c | c c}
            & $r_{21}$ & $r_{22}$ \\
            \hline
            1 & 170 & 170 \\
            2 & 170 & 170 \\
            3 & 170 & 170 \\
            \end{tabular}
        }
        \parbox{.15\linewidth}{
            \centering
            \begin{tabular}{c | c}
            & $r_{31}$\\
           \hline
            1 & -1360 \\
            2 & -280 \\
            3 & -280 \\
            \end{tabular}
        }

    \end{example}

    Consider any solution $(A,P,J)$ where $J \ne \{1\}$. This solution then only assigns players to rooms of type $2$ and $3$, and therefore the total utility of all three players under this solution must be at most $60$. This means that the three players would want to deviate together to the apartment $1$, and therefore this solution is not in the core. Therefore, any solution in the core must be of the form $(A,P,\{1\})$. Consider such a solution. Because the rooms in apartment $1$ are symmetric, WLOG assume that $A(i) = r_{1i}$ for all players $i$. Since the total rent is $0$, we must have
    \[
        P(r_{11}) + P(r_{12}) + P(r_{13}) = 0.
    \]
    Let $u_1,u_2,u_3$ be the utilities for players $1,2,3$ under solution $(A,P,\{1\})$. We know that $u_1 = 340 - P(r_{11})$, $u_2 = -20 - P(r_{12})$, and $u_3 = -20 - P(r_{13})$. Then $u_1,u_2,u_3$ must satisfy the following system of equations if the solution $(A,P,\{1\})$ is in the core.
    \[
        u_1 + u_2 + u_3 = 300
    \]
    \[
        u_1 + u_2 \ge 340
    \]
    \[
        u_1 + u_3 \ge 340
    \]
    \[
        u_2 + u_3 \ge 340
    \]

    However, this system of equations has no solution. Therefore no such solution exists, and the core is empty.
\end{proof}

\section{Universal Envy-Freeness for correlated distributions}\label{app:uef_corr}

In this section we will look at the existence of a universal envy-free solution in the setting of $m$ apartments and two players who have values drawn from a Bernoulli distribution. First, we need the following lemma which exactly characterizes the event that a universal envy-free solution exists when values of both players are $0$ or $1$. For simplicity, we will assume that the rent of every apartment is equal to $1$.
\begin{lemma}\label{lemma:corr}
    Suppose we have $m$ apartments each with rent $1$ and two players such that $V_i(r) \in \{0,1\}$ for every room $r$. Then there does not exist a universal envy-free solution if and only if all of the following three events occur.
    \begin{itemize}
        \item No apartment has an assignment where both players value their assigned rooms with value $1$. Denote this event as $E_1$.
        \item Both players have value $1$ for at least one room. Denote this event as $E_2$.
        \item In at least one apartment, one player has value $1$ for both rooms and the other player has value $0$ for both rooms. Denote this event as $E_3$.
    \end{itemize}
\end{lemma}
\begin{proof}
    First we will show the ``only if" direction. To show this we will prove the contrapositive, which is that if any one of these three events does not hold, then there does exist a universal envy-free solution. Assume $E_1$ does not hold. Then there must be an apartment $j$ and an corresponding assignment $A_j^*$ such that $V_1(A_j^*(1)) = V_2(A_j^*(2)) = 1$. Let $A$ be any assignment such that $A_j = A_j^*$, and let $P$ be the price matrix with every entry equal to $0.5$. Then $(A,P,j)$ will satisfy universal envy-freeness. This is because no player can have utility of more than $0.5$ for any room in any apartment, and both players have utility of exactly $0.5$ for their assigned rooms in apartment $j$. 

    Assume instead that event $E_2$ does not hold. Then there must be some player who has utility $0$ for every room. WLOG let this player be player $1$. Choose $A$ to be a welfare-maximizing assignment in every apartment and choose $P$ to be a price matrix with every entry equal to $0.5$. Let $j$ be the apartment in which player $2$ has the highest utility.
    Then $(A,P,j)$ is a universal envy-free solution. This is because the utility of player $1$ for every room is $-0.5$, and we chose $j$ such that $U_2(A_j, P) > U_2(A_{j'}, P)$ for $j' \neq j$.

    Finally, assume that event $E_3$ does not hold. Then no apartment has one player with value $1$ for both rooms and the other player with value $0$ for both rooms. If $E_1$ does not hold, then we already have shown that a universal envy-free solution exists. Therefore, assume that $E_1$ holds and $E_3$ does not hold. There are only three possible sets of values for an apartment under $E_1 \cap \neg E_3$. These possibilities are: 1. both players have value $0$ for every room in that apartment, 2. Exactly one player has value $1$ for exactly one room in that apartment, 3. Both players have value $1$ for the same room in that apartment. Construct the price matrix $P$ as follows. For any apartment in case $1$, assign price $0.5$ to both rooms, and for any apartment in cases $2$ and $3$, assign price $1$ to the room that has a non-0 value by at least one of the players. Under these prices, both players will have utility exactly $0$ for both rooms in an apartment satisfying case $3$. Similarly, both players will have utility exactly $-0.5$ for both rooms of any apartment satisfying case $1$. Finally, for any apartment satisfying case $2$, for any maximum welfare assignment both players will have utility $0$ for their assigned rooms and utility at most $0$ for the other room. Choose $A$ to be a welfare-maximizing assignment in every apartment and let $j$ be an apartment which has the highest utility under the partial solution $(A,P)$. Then $(A,P,j)$ is a universal envy-free solution. This completes the proof of the ``only if" direction.

    \vspace{5mm}

    Now we will show the ``if" direction. We want to show that if all of events $E_1, E_2, E_3$ hold, then there is no universal envy-free solution. Proof by contradiction. Assume events $E_1,E_2,E_3$ all hold and there exists a universal envy-free solution $(A,P,j^*)$. Under event $E_1$, there is no possible price assignment and choice of apartment where the sum of the player utilities exceeds $0$. Therefore, in apartment $j^*$, both players have utility at most $0$. Suppose apartment $j$ satisfies event $E_3$, and player $1$ has value $1$ for both rooms in apartment $j$. Then under any prices, player $1$ has utility of at least $0.5$ for one room in apartment $j$. Therefore, since $(A,P,j^*)$ satisfies universal envy-freeness, player $1$ has utility at least $0.5$ in apartment $j^*$. Since the total utility in apartment $j^*$ is at most $0$ under event $E_1$, this implies that player $2$ has utility at most $-0.5$ in apartment $j^*$. However, under event $E_2$, player $2$ must have value $1$ for at least one room in some apartment $j'$. This implies that player $2$ has utility of at least $0$ for at least one room in apartment $j'$ under partial solution $(A,P)$. Since $(A,P,j^*)$ satisfies universal envy-freeness, this implies that player $2$ must have utility of at least $0$ in apartment $j^*$. This is a contradiction.
\end{proof}

\begin{lemma}\label{lemma:corr_result}
    Suppose there are $m$ apartments and two players, and fix $r \in [0,1]$. For every room $r_{jk}$ in any apartment $j$, suppose $V_{1}(r_{jk})$ and $V_{2}(r_{jk})$ are drawn from a joint distribution such that the marginal distributions of both $V_1(r_{jk})$ and $V_2(r_{jk})$ are both $Bernoulli(1/2)$ and $\Pr\left(V_1(r_{jk}) = V_2(r_{jk})\right) = r$. Note that the values are drawn independently for every room. Then
    {\small
    \[
        \Pr(\text{Exists a universal envy-free solution}) =  
        1 - \left(\left(\frac{r^2 + 2}{4}\right)^m - 2\left(\frac{1}{4}\right)^m - \left(\frac{4r - r^2}{4}\right)^m + 2\left(\frac{2r -r^2 }{4}\right)^m\right)
    \]
    }
\end{lemma}
\begin{proof}
    We will calculate the probability of $E_1 \cap E_2 \cap E_3$ from Lemma \ref{lemma:corr}. We have four cases for utilities within an apartment if event $E_1$ holds.

    \begin{itemize}
        \item One player has value $1$ for both rooms and the other player has value $0$ for both rooms. There are two symmetric ways to achieve this each with probability $p_1 = \left(\frac{1}{4}(1-r)^2\right)$. 
        \item Every player has value $0$ for every room. There is one way to achieve this with $p_2 = \left(\frac{1}{4}r^2\right)$.
        \item Exactly one player has value $1$ for exactly one room in the apartment. There are four symmetric ways to achieve this, each with probability $p_3 = \left(\frac{1}{4}(1-r)r\right)$.
        \item Both players have value $1$ for the same room and both players have value $0$ for the other room. There are two symmetric ways to achieve this, each with probability $p_4 = \left(\frac{1}{4}r^2\right)$.
    \end{itemize}

    Using these four cases, we have that
    \[
        \Pr(E_1) = (2p_1 + p_2 + 4p_3 + 2p_4)^m = \left(\frac{r^2 + 2}{4}\right)^m
    \]
    \[
        \Pr(\neg E_2 \cap E_1) = 2(p_1 + p_2 + 2p_3)^m - p_2^m = 2\left(\frac{1}{4}\right)^m - \left(\frac{r^2}{4}\right)^m
    \]
    \[
        \Pr(\neg E_3 \cap E_1) = (p_2 + 4p_3 + 2p_4)^m = \left(\frac{4r - r^2}{4}\right)^m
    \]
    \[
        \Pr( \neg E_2 \cap \neg E_3 \cap E_1) = 2(p_2 + 2p_3)^m - p_2^m = 2\left(\frac{2r -r^2 }{4}\right)^m - \left(\frac{r^2}{4}\right)^m
    \]
    Putting this all together, by Lemma \ref{lemma:corr}, we have that
    {\small
    \begin{align*}
        \Pr(\text{No universal envy-free solution}) &= \Pr(E_1) -   \Pr(\neg E_2 \cap E_1) - \Pr(\neg E_3 \cap E_1) +  \Pr( \neg E_2 \cap \neg E_3 \cap E_1) \\
        &= \left(\frac{r^2 + 2}{4}\right)^m - 2\left(\frac{1}{4}\right)^m - \left(\frac{4r - r^2}{4}\right)^m + 2\left(\frac{2r -r^2 }{4}\right)^m
    \end{align*}
    }
    Lemma \ref{lemma:corr_result} exactly characterizes the probability of the existence of a university envy-free solution in terms of the dependency $r$ between the player values. Interestingly, we note that for sufficiently large $m$ (for example $m = 10$), this function is not monotonically decreasing and has a global minimum for $r$ strictly between $0$ and $1$.
    
   \end{proof}

\newpage
\section{Strong Negotiated Envy-Freeness}\label{sec:strong_REF}

\subsection{Definition and Main Result}

We begin with an illustrative example. Consider the following instance, for which there does not exist a universal envy-free solution.

\begin{example}\label{example:EEF}
    There are two players and two apartments. Each apartment has total rent $100$ and contains two symmetric rooms. The value of each player for each room is shown below:

    \parbox{.45\linewidth}{
        \centering
        \begin{tabular}{c | c c}
        & $r_{11}$ & $r_{12}$ \\
        \hline
        1 & 100 & 100 \\
        2 & 1 &   1  \\
        \end{tabular}
        }
        \parbox{.45\linewidth}{
        \centering
        \begin{tabular}{c | c c}
        & $r_{21}$ & $r_{22}$ \\
        \hline
        1 & 0 & 0 \\
        2 & 99 &   99  \\
        \end{tabular}
        }
\end{example}

As the rooms are symmetric within each apartment, we can WLOG consider only the assignment $A$ where player 1 is assigned room 1 in both apartments. We can furthermore conclude that the only price matrix which satisfies individually envy-freeness assigns price $50$ to every room. By Lemma \ref{lemma:highest_max_welfare}, only apartment $1$ can be a consensus apartment. Therefore, through simple calculations, we can observe that the only price matrices which satisfy negotiated envy-freeness are of the form $P(A_1(1)) = 99+x$, $P(A_2(1)) = 1-x$, $P(A_2(1)) = 1-x$, $P(A_2(2)) = 99+x$, for $x \in [0,1]$. 

In every such solution, player 1 is envious of player 2 in the consensus apartment, and therefore no solution of this form will satisfy universal envy-freeness. However, we would like to provide a justification for why this solution is fair even though there exists envy in the consensus apartment. Consider the sequence of negotiations necessary to move from the individually envy-free price matrix where $P(A_1(1)) =P(A_1(2)) =P(A_2(1)) =P(A_2(2)) = 50$ to a negotiated envy-free price matrix as characterized above. The negotiations must have consisted of player 1 increasing their rent in apartment 1 by a total of $49 + x$ and decreasing their rent in apartment 2 by the same amount. In the negotiation interpretation, player 1 is willing to negotiate only to reach consensus with player 2. Therefore, player 1 would want to stop negotiating once player 2 weakly prefers apartment 1. In other words, a ``fair" stopping point for the negotiations would be exactly when consensus is reached. Using this ``fair" stopping point is equivalent to choosing $x = 0$ in the solution set described above. When $x=0$, the envy of player 1 in the consensus apartment can be explained as the result of negotiations that were necessary to reach consensus with player 2.

In this example, the envy of player 1 in the consensus apartment is caused by the rent decrease of player 2 and corresponding rent increase of player 1 in the consensus apartment. Using this intuition, we can impose a cap on the rent decrease any player can have in the consensus apartment in any instance. For any player $i$, apartment $j^*$, and partial solution $(A,P)$, define $S_i(A,P,j^*) \subseteq [m]$ as the set of apartments that player $i$ strictly prefers to apartment $j^*$ under the partial solution $(A,P)$. Equivalently, $S_i(A,P, j^*) = \{j : U_i(A_{j}, P) > U_i(A_{j^*}, P)\}$. 
\begin{lemma}\label{lemma:delta_quantity}
    Starting from partial solution $(A,P)$, let $\{\tau_t\}_{t=1}^T = \{(\delta^t, i,i^t, j^t, j^*)\}_{t=1}^T$ be a sequence of $T$ negotiations that each decrease player $i$'s price in apartment $j^*$ and suppose that after the $T$th negotiation, player $i$ prefers apartment $j^*$ to any other apartment. Then 
    \begin{equation}\label{eq:delta_quantity}
        \min_{T,\{\tau_t\}_{t=1}^T} \sum_{t=1}^T \delta^t =  \frac{1}{|S_i(A,P, j^*) | + 1} \cdot \sum_{j \in S_i(A,P, j^*) } \left( U_i(A_{j}, P) -U_i(A_{j^*}, P)\right).
    \end{equation}
\end{lemma}
\begin{proof}
    Define  $\Delta = \sum_{t=1}^T \delta^t$. Then the $T$ negotiations will decrease player $i$'s price in apartment $j^*$ by a total of $\Delta$, thereby increasing player $i$'s utility in apartment $j^*$ by $\Delta$. For every apartment $j \ne j^*$, define $\Delta_j$ as the amount negotiated with that apartment among the $T$ negotiations, which implies that $\sum_{j \ne j^*} \Delta_j = \Delta$. In words, in every apartment $j \ne j^*$, player $i$'s utility decreases by $\Delta_j$. Assume that the $T$ negotiations resulted in player $i$ weakly preferring apartment $j^*$ to every other apartment. This implies that, for every $j \ne j^*$, 
    \begin{equation}\label{eq:trading}
        U_i(A_{j^*}, P) + \Delta \ge U_i(A_{j}, P) - \Delta_j,
    \end{equation}
    which is equivalent to the condition that $\Delta + \Delta_j \ge U_i(A_{j}, P) - U_i(A_{j^*}, P)$ for every $j\ne j^*$. Summing this for all $j$ in $S_i(A,P,j^*)$, we have that
    \[
        |S_i(A,P,j^*)| \cdot \Delta + \sum_{j \in S_i(A,P,j^*)} \Delta_j \ge  \sum_{j \in S_i(A,P,j^*)} \left(U_i(A_{j}, P) - U_i(A_{j^*}, P)\right).
    \]
    Using that $\sum_{j \in S_i(A,P,j^*)} \Delta_j \le \Delta$, we conclude that
    \[
        \Delta \ge \frac{1}{|S_i(A,P,j^*)| + 1} \cdot  \sum_{j \in S_i(A,P,j^*)}\left( U_i(A_{j}, P) - U_i(A_{j^*}, P)\right).
    \]    
    We have shown the $\ge$ direction of the desired result. Now we want to show that there exist a sequence of negotiations that sum to the desired value of $\Delta$. This is equivalent to finding values of $\Delta_j$ that satisfy Equation \eqref{eq:trading} for all $j$ and sum to the value on the right hand side of Equation \eqref{eq:delta_quantity}. Let 
    \begin{equation}\label{eq:delta}
    \Delta_j = U_i(A_{j}, P) - U_i(A_{j^*}, P) - \frac{1}{|S_i(A,P,j^*)| + 1}  \cdot \sum_{j \in S_i(A,P,j^*)} \left(U_i(A_{j}, P) - U_i(A_{j^*}, P)\right)
    \end{equation}
    for every $j \in S(A,P,j^*)$ and $\Delta_j = 0$ otherwise. Then 
    \begin{align*}
        \Delta &= \sum_{j=1} \Delta_j = \frac{1}{|S_i(A,P,j^*)| + 1} \cdot  \sum_{j \in S_i(A,P,j^*)} \left(U_i(A_{j}, P) - U_i(A_{j^*}, P)\right)
    \end{align*}
    as desired. Plugging this expression for $\Delta$ into Equation \eqref{eq:delta} gives that Equation \eqref{eq:trading} is satisfied for all $j \in S_i(A,P,j^*)$. Equation \eqref{eq:trading} is satisfied for all $j \not\in S_i(A,P,j^*)$ because by definition these $j$ satisfy $U_i(A_{j^*},P) \ge U_i(A_j,P)$. Therefore, a sequence of negotiations that achieves the desired minimum value of $\Delta$ is $\{(\Delta_{t}, i,i^t, t, j^*)\}_{t=1}^m$ for any choice of $i^t$.
\end{proof}
In Definition \ref{def:strong_REF}, we define a stronger version of negotiated envy-freeness such that no player's price in the consensus apartment decreases by more than the quantity specified in Lemma \ref{lemma:delta_quantity}.
\begin{definition}\label{def:strong_REF}
    A solution $(A,P,j^*)$ satisfies \textbf{strong negotiated envy-freeness} if 
    \begin{enumerate}
        \item $(A,P,j^*)$ satisfies consensus.
        \item $\exists$ a price matrix $Q$ such that $(A,Q)$ is individually envy-free and Equations \ref{eq:strong_REF_equality} and \ref{eq:strong_REF} hold $\forall i$.
    \end{enumerate}
    \begin{equation}\label{eq:strong_REF_equality}
        \sum_{j = 1}^m P(A_j(i)) = \sum_{j = 1}^m Q(A_j(i))
    \end{equation}
    \begin{equation}\label{eq:strong_REF}
        P(A_{j^*}(i)) \ge Q(A_{j^*}(i)) - \frac{1}{|S_i(A,Q, j^*) | + 1}\cdot \sum_{j \in S_i(A,Q, j^*) } \left(U_i(A_{j}(i), Q) -U_i(A_{j^*}(i), Q)\right).
    \end{equation}
\end{definition}

Note that a direct result of a solution satisfying strong negotiated envy-freeness is that within the consensus apartment $j^*$, no player is envious of any other player by more than \[
2\sum_i \left( \frac{1}{|S_i(A,Q, j^*) | + 1}\cdot\sum_{j \in S_i(A,Q, j^*) } \left(U_i(A_{j}(i), Q) -U_i(A_{j^*}(i), Q)\right)\right).
\]
This upper bound is directly related to how much each player preferred their assignment in other apartments over their assignment in the consensus apartment in the corresponding individually envy-free solution.

Strong negotiated envy-free solutions are a subset of negotiated envy-free solutions, with the advantage that the envy in the consensus apartment of strong negotiated envy-free solutions has an upper bound that is justifiable by necessary negotiations. As the following theorem shows, this benefit comes at no cost in terms of our positive results. 

\begin{thm}\label{thm:explainable}
    There always exists a solution $(A,P,j^*)$ that satisfies strong negotiated envy-freeness. Furthermore, optimizing an objective as in Theorem \ref{thm:max_obj} subject to strong negotiated envy-freeness can be done in time polynomial in both $n$ and $m$. 
\end{thm}
\begin{proof}[Proof sketch]
We provide a constructive proof of existence by presenting a rebalancing algorithm. This algorithm starts with a partial solution $(A,Q)$ that is individually envy-free, and after a sequence of negotiations is guaranteed to terminate with a solution $(A^*,P,j^*)$ that satisfies strong negotiated envy-freeness. Without loss of generality, assume that apartment $1$ is a welfare-maximizing apartment. Informally, the algorithm works as follows. The algorithm iterates over every apartment $j=1,2,...,m$, with the goal of maintaining the invariant that after the $j$th iteration, apartment $1$ is the consensus apartment among the first $j$ apartments. In the body of the $j$th iteration, the goal is to ``rebalance" the prices in apartment $j$ so that no player prefers apartment $j$ to apartment $1$. Any player $i$ who currently prefers apartment $j$ to apartment $1$ by $\Delta$ will have her rent increased by $\frac{m-1}{m}\Delta$ in apartment $j$ and her rent decreased evenly in every other apartment by $\frac{\Delta}{m}$. In order for player $i$ to have her rent changed in this way, she must negotiate with at least one other player $i'$ who will have his rent decreased in apartment $j$ and increased in every other apartment. Specifically, player $i$ will negotiate with the set $T$ of players that currently prefer apartment $1$ to apartment $j$. The complication in the algorithm comes from deciding how much player $i$ can negotiate with each of the players in $T$, as no player in $T$ can have his rent increase drastically enough in apartment $1$ that he no longer prefers apartment $1$ to each of the first $j$ apartments. The formal algorithm and proof are presented in Appendix \ref{app:strong_REF_alg}.

To see that we can optimize objectives under strong negotiated envy-freeness, note that Equation \eqref{eq:strong_REF} can be expressed as a linear constraint in the linear program \ref{lp:1}. Therefore, optimizing an objective with respect to all strong negotiated envy-free solutions can still be done in polynomial time as in \ref{thm:max_obj} with this additional linear constraint added to the linear program. 
\end{proof}

\subsection{Proof of Theorem \ref{thm:explainable}}\label{app:strong_REF_alg}

\begin{proof}

To prove Theorem \ref{thm:explainable}, we will show that Algorithm \ref{algo:EEF} terminates in a polynomial number of steps, and the solution returned by the algorithm satisfies strong negotiated envy-freeness.

\begin{algorithm}[htb!]
\caption{[Strong Negotiated Envy-Freeness]}
\label{algo:EEF}
\begin{algorithmic}[1]
\Require $(A,Q)$ is individually envy-free.
\State $P \gets Q$
\For{$j \leftarrow 1$ to $m$} \label{line:outerfor} 
    \LineComment{Negotiate until every player prefers apartment $1$ to apartment $j$}
    \While{$\left|\{i : U_i(A_{j}(i),P) > U_i(A_1, P)\}\right| > 0$} \label{line:firstwhile} 
        \State $i \gets \min \{i : U_i(A_{j}(i),P) > U_i(A_1, P)\}$
        \State $\Delta = U_i(A_j, P) - U_i(A_1, P)$
        \State $P(A_j(i)) = P(A_j(i)) + \frac{(m-1)\Delta}{m}$ \label{line:dp}
        \State $\forall j' \ne j : P(A_{j'}(i)) = P(A_{j'}(i)) - \frac{\Delta}{m}$ \label{line:incr_util}
        \LineComment{Redistribute the price changes among the other players evenly}
        \While{$\Delta > 0$} \label{line:secondwhile}
            \State $T \gets \{i : U_i(A_{j}(i),P) < U_i(A_1, P)\} $
            \State $\epsilon \gets \min_{i' \in T} \left| U_i(A_1, P) - U_i(A_{j}(i),P)\right|$
            \If{$\epsilon \ge \frac{\Delta}{|T|}$}
                \For{$i' \in T$}
                    \State $P(A_j(i')) = P(A_j(i')) - \frac{(m-1)\Delta}{m|T|}$
                    \State $\forall j' \ne j : P(A_{j'}(i')) = P(A_{j'}(i')) + \frac{\Delta}{m|T|}$
                \EndFor
                \State $\Delta = 0$
            \Else
                \For{$i' \in T$}
                    \State $P(A_j(i')) = P(A_j(i')) - \frac{(m-1)\epsilon}{m}$
                     \State $\forall j' \ne j : P(A_{j'}(i')) = P(A_{j'}(i')) + \frac{\epsilon}{m}$
                \EndFor
                \State $\Delta = \Delta - |T|\epsilon$
            \EndIf
        \EndWhile
    \EndWhile
\EndFor \\
\Return $(A, P, 1)$ 
\end{algorithmic}
\end{algorithm}

    We will first argue that the algorithm finishes in polynomial time. The outer ``for" loop on Line \ref{line:outerfor} goes through $m$ iterations. After each iteration of the ``while" loop in Line \ref{line:firstwhile}, the size of the set $\left|\{i : U_i(A_{j}(i),P) > U_i(A_1, P)\}\right|$ decreases by $1$, and therefore there are at most $n$ iterations of this loop per iteration of the outer ``for'' loop. In each iteration of the ``while" loop on Line \ref{line:secondwhile}, the size of the set $T$ decreases by at least $1$. Therefore, since the size of the set $T$ is at most $n$, the inner while loop has at most $n$ iterations. Therefore, we can conclude that the entire algorithm runs in polynomial time $O(n^2m)$ and always terminates.

    After the $j$th iteration of the ``for" loop on Line \ref{line:outerfor}, apartment $1$ is weakly preferred among the first $j$ apartments by all $n$ players, as this is exactly the condition for terminating the while loop on Line \ref{line:firstwhile}. Therefore, after the ``for" loop on Line \ref{line:outerfor} completes $m$ iterations, the first apartment will be a consensus apartment. The algorithm is also constructed so that no player's total utility for their assigned $m$ rooms ever changes. Since $(A,Q)$ begins as an individually envy-free assignment, this implies that the returned solution $(A,P,1)$ satisfies negotiated envy-freeness.

    In order to show that Algorithm \ref{algo:EEF} further satisfies strong negotiated envy-freeness, we need to show that Equation \ref{eq:strong_REF} holds for all $i$, i.e.
    \[
        P(A_{1}(i)) \ge Q(A_{1}(i)) - \frac{1}{|S_i(A,Q, 1) | + 1}\cdot \sum_{j \in S_i(A,Q, 1) } \left(U_i(A_{j}, Q) -U_i(A_{1}, Q)\right) \quad \forall i
    \]
    Intuitively, Equation \ref{eq:strong_REF} requires that player $i$'s rent in apartment $1$ is not decreased by too much when moving from price matrix $Q$ to price matrix $P$. The only part of the algorithm where $P(A_1(i))$ can decrease is in Line \ref{line:incr_util}, and such a decrease only occurs in round $j$ if player $i$ strictly preferred apartment $j$ to apartment $1$ under $(A,Q)$. In this case, player $i$ will have their rent in apartment $1$ decrease by $\frac{U_i(A_j, P) - U_i(A_1, P)}{m} = \frac{U_i(A_j, Q) - U_i(A_1, Q)}{m} $. This equality is because until Line \ref{line:incr_util} is reached for apartment $j$ and player $i$, we have $U_i(A_1, P) = U_i(A_1, Q)$ and $U_i(A_j, P) = U_i(A_j, Q)$.
    
    Therefore, we can conclude that
    \begin{align*}
        P(A_1(i)) &\ge Q(A_1(i)) - \sum_{j : U_i(A_j, Q) > U_i(A_1, Q)}  \frac{U_i(A_j, Q) - U_i(A_1, Q)}{m} \\
        &= Q(A_1(i)) - \frac{1}{m}\sum_{j \in S_i(A,Q, 1) } U_i(A_{j}, Q) -U_i(A_{1}, Q) \\
        &\ge Q(A_1(i)) - \frac{1}{|S_i(A,Q, 1) | + 1}\sum_{j \in S_i(A,Q, 1) } U_i(A_{j}, Q) -U_i(A_{1}, Q)
    \end{align*}
    where the last inequality is because $m \ge |S_i(A,Q, 1) | + 1$.
\end{proof}

\end{document}